\ifx\documentclass\undefined
\documentstyle[12pt]{article}
\else
\documentclass[12pt]{article}
\fi

\usepackage{color}

\usepackage{ascmac}
\usepackage{amsmath}
\usepackage{amssymb}

\usepackage{here}

\usepackage[dvipdfmx]{graphicx}
\usepackage{pdfpages}
\usepackage{amsmath}

\newcommand{\beq}{\begin{eqnarray*}}
\newcommand{\eeq}{\end{eqnarray*}}
\makeatletter
\renewcommand{\theequation}{\thesection.\arabic{equation}}
\@addtoreset{equation}{section}
\def\eqnarray{%
\stepcounter{equation}%
\let\@currentlabel=\theequation
\global\@eqnswtrue
\global\@eqcnt\z@
\tabskip\@centering
\let\\=\@eqncr
$$\halign to \displaywidth\bgroup\@eqnsel\hskip\@centering
$\displaystyle\tabskip\z@{##}$&\global\@eqcnt\@ne
\hfil$\displaystyle{{}##{}}$\hfil
&\global\@eqcnt\tw@$\displaystyle\tabskip\z@{##}$\hfil
\tabskip\@centering&\llap{##}\tabskip\z@\cr}
\makeatother
\newtheorem{theorem}{Theorem}[section]
\newtheorem{lemma}[theorem]{Lemma}
\newtheorem{corollary}[theorem]{Corollary}
\newtheorem{proposition}[theorem]{Proposition}

\newsavebox{\toy}
\savebox{\toy}{\framebox[0.65em]{\rule{0cm}{1ex}}}
\newcommand{\QED}{\usebox{\toy}}
\def\nlni{\par\ifvmode\removelastskip\fi\vskip\baselineskip\noindent}
\newenvironment{proof}{\nlni\begingroup\it Proof.\rm}{
\endgroup\vskip\baselineskip}

\begin{document}
\setlength{\baselineskip}{15pt}
\title{
Shape of eigenvectors 
for the decaying potential model\\
}
\author{
Fumihiko Nakano
\thanks{
Mathematical Institute,
Tohoku University,
Sendai 980-8578, Japan
e-mail : 
fumihiko.nakano.e4@tohoku.ac.jp
}
}
\maketitle
\begin{abstract}
We 
consider the 1d Schr\"odinger operator with decaying random potential, and study the joint scaling limit of the eigenvalues and the measures associated with the corresponding eigenfunctions which is based on the formulation by Rifkind-Virag 
\cite{RV}.
As a result, 
we have 
completely different behavior depending on the decaying rate
$\alpha > 0$
of the potential : the limiting measure is equal to 
(1)
Lebesgue measure for the 
super-critical case 
($\alpha > 1/2$),  
(2)
a measure of which the density has power-law decay with Brownian fluctuation for critical case
($\alpha=1/2$), 
and 
(3)
the delta measure with its atom being uniformly distributed for 
the sub-critical case
($\alpha<1/2$).
This result is 
consistent with previous study on spectral and statistical properties. 
\end{abstract}


\section{Introduction}
In this paper
we consider the one-dimensional Schr\"odinger operator with random decaying potential : 
\beq
H := 
- \frac {d^2}{dt^2} + a(t) F(X_t) 
\eeq
where 
$a \in C^{\infty}({\bf R})$, 
$a(-t) = a(t)$, 
$a(t)$ 
is monotone decreasing for 
$t>0$ 
and 
\beq
a(t) = t^{- \alpha} (1 + o(1)), 
\quad
t \to \infty
\eeq
for some 
$\alpha > 0$.
For 
$\alpha < 1/2$, 
we also need 
$a'(t) = {\cal O}(t^{- \alpha - 1})$
to use the results in 
\cite{KN2}.
$F \in C^{\infty}(M)$
is a smooth function on a torus 
$M$
such that 
\beq
\langle F \rangle
:=
\int_M F(x) dx = 0
\eeq
and 
$\{ X_t \}_{t \in {\bf R}}$
is the Brownian motion on 
$M$.
Since 
$a(t) F(X_t)$
is a compact perturbation with respect to 
$(-\triangle)$, 
the spectrum  
$\sigma (H) \cap (-\infty, 0)$
on the negative real axis is discrete.
The spectrum 
$\sigma (H) \cap [0, \infty)$ 
on the positive real axis is 
\cite{KU} : 
\beq
\sigma (H) \cap [0, \infty)
\mbox{ is }
\left\{
\begin{array}{lc}
\mbox{ a.c. } & (\alpha>1/2) \\
\mbox{ p.p. on }[0, E_c] \mbox{ and s.c. on }[E_c, \infty) 
& (\alpha=1/2) \\
\mbox{ p.p. } & (\alpha < 1/2) \\
\end{array}
\right.
\eeq
where 
$E_c$ 
is a deterministic constant.
In fact, 
it is shown in 
\cite{KU} 
that the generalized eigenfunctions of 
$H$
are bounded for super-critical case($\alpha > 1/2$), 
have power-law decay for critical case($\alpha=1/2$) 
and are sub-exponentially localized for sub-critical case($\alpha<1/2$). 
For the 
level statistics problem, 
we consider the point process 
$\xi_{n, E_0}$ 
composed of the rescaling eigenvalues 
$\{ n(\sqrt{E_j (n)} - \sqrt{E_0}) \}_j$ 
of the finite box Dirichlet Hamiltonian 
$H_n := H |_{[0, n]}$
around the reference energy 
$E_0 > 0$,  
whose behavior as 
$n \to \infty$
is given by 
\cite{KN1, N2, KN2}
\beq
\xi_{n, E_0} 
\stackrel{d}{\to}
\left\{
\begin{array}{lc}
Clock(\theta(E_0)) & (\alpha>1/2) \\
Sine(\beta(E_0)) & (\alpha=1/2) \\
Poisson(d \lambda / \pi) & (\alpha<1/2)
\end{array}
\right.
\eeq
where 
$Clock (\theta) := 
\sum_{ n \in {\bf Z}} \delta_{n \pi + \theta}$, 
is the clock process for a random variable 
$\theta$
on 
$[0, \pi)$,
$Sine(\beta)$ 
is the Sine$_{\beta}$-process which is the bulk scaling limit of the Gaussian beta emsemble \cite{VV}, and 
Poisson ($\mu$) is a Poisson process on 
${\bf R}$
with intensity measure 
$\mu$.
$\theta (E_0)$
has a form of the projection onto the torus 
$[0, \pi)$ 
of a time-change of a Brownian motion \cite{KN1}. 
$\beta (E_0) = \tau(E_0)^{-1}$ 
is equal to the reciprocal number of 
$\tau(E_0)$ 
explicit form of which is given in (\ref{tau}).
$\tau(E_0)$ 
is the ``Lyapunov exponent"
such that the solution to the Schr\"odinger equation 
$H \varphi = E \varphi$ 
has the power-law decay : 
$\varphi(x) \simeq |x|^{- \tau (E)}$, 
$|x| \to \infty$.
Since
$\lim_{E_0 \downarrow 0} \beta(E_0) = 0$ 
and 
$\lim_{E_0 \uparrow \infty} \beta (E_0) = \infty$, 
repulsion of eigenvalues near 
$E_0$ 
is small (resp. large) 
if 
$E_0$ 
is small (resp. large), which is consistent with the following fact \cite{AD, N3} : 
\begin{equation}
Sine (\beta) 
\stackrel{d}{\to} 
\left\{
\begin{array}{ll}
Poisson (d \lambda / \pi) & (\beta \downarrow 0) \\
Clock(unif[0, \pi)) & (\beta \uparrow \infty)
\end{array}
\right.
\label{Sinebetalimit}
\end{equation}
In this paper, 
we consider the scaling limit of the measure corresponding to the eigenfunction of 
$H_L$
along  the formulation studied by Rifkind-Virag 
\cite{RV}.
%
%
Let 
$\{ E_j (n) \}_{j\ge 1}$ 
be the positive eigenvalues of 
$H_n$, 
and let
$\{ \psi^{(n)}_{E_j(n)} \}$ 
be the corresponding eigenfunctions. 
We consider the associated random probability measure
$\mu^{(n)}_{E_j(n)}$ 
on 
$(0,1)$.
%
\begin{equation}
\mu^{(n)}_{E_j(n)} (dt)
:=
Cn
\left(
\left| 
\psi^{(n)}_{E_j(n)}(nt) 
\right|^2 
+
\frac {1}{E_j(n)}
\left| 
\frac {d}{dt}
\psi^{(n)}_{E_j(n)} (nt)
\right|^2
\right) 
dt
\label{randommeasure}
\end{equation}
where 
$C$
is the normalizing constant.
In 
(\ref{randommeasure}), 
we consider the derivative of the 
$\psi^{(n)}_{E_j(n)}$ 
as well as 
$\psi^{(n)}_{E_j(n)}$
so that the analysis is reduced to that of the radial part of the Pr\"ufer variable to be introduced in (\ref{Prufer}).
However, 
since 
$\psi$
and 
$\psi'$
have same global behavior, 
and since 
we are interested in the global shape of 
$\psi$, 
we believe that 
(\ref{randommeasure})
is a reasonable definition to study the shape of eigenvectors. 
Let 
$J := [a,b] (\subset (0, \infty))$ 
be an interval, 
${\cal E}^{(n)}_J := 
\{ E_j (n) \}_j \cap J$ 
be the set of eigenvalues of 
$H_n$ 
in 
$J$, 
and 
$E_J^{(n)}$ 
be the random variable uniformly distributed on 
${\cal E}^{(n)}_J$.
Our aim is 
to consider the large 
$n$ 
limit of the joint distribution of the eigenvalue-eigenvector pairs : 
\beq
{\bf Q} : 
\left(
E_J^{(n)}, \mu_{E_J^{(n)}}^{(n)}
\right)
\;
\stackrel{d}{\to} 
\;
?
\eeq
For 
d-dimensional discrete random Schr\"odinger operator, 
if 
$J$ 
is in the localized region, we have
\cite{N1}
\begin{equation}
\left(
E_J^{(n)}, \mu_{E_J^{(n)}}^{(n)}
\right)
\stackrel{d}{\to}
\left(
E_J, \delta_{unif[0, 1]^d}
\right)
\label{previousresult}
\end{equation}
where 
$E_J$ 
is the random variable whose distribution is equal to  
$N(J)^{-1} 1_J (E) d N(E)$, 
where 
$dN$
is the density of states measure.
Rifkind-Virag \cite{RV} 
studied the 1-d discrete Schr\"odinger operator with critical
($\alpha = 1/2$) 
decaying coupling constant, and obtained that  
the limit of 
$\mu^{(n)}_{E_J^{(n)}}$
is given by an exponential Brownian motion with negative drift which corresponds to the exponential decay of the eigenfunctions. 

~
To state our result, we need notations further.
Let 
$N(E) := \pi^{-1} \sqrt{E}$
be the integrated density of states of 
$H$, 
$N(J) := N(b) - N(a)$, 
and let 
\begin{equation}
\tau(E)
:= 
\frac {1}{8E}
\int_M | \nabla (L + 2i \sqrt{E})^{-1} F|^2 dx
\label{tau}
\end{equation}
where 
$L$ 
is the generator of 
$(X_t)$. 
Moreover, let 
$E_J$
be the random variable whose distribution is equal to 
$N(J)^{-1} 1_J (E) d N(E)$, 
let 
$U$ 
be the uniform distribution on 
$(0,1)$, 
and let 
${\cal Z}$ 
be the 2-sided Brownian motion, where 
$E_J$, $U$, and 
${\cal Z}$
are independent.
\begin{theorem}
\beq
&&
\left(
E_J^{(n)}, \mu^{(n)}_{E_J^{(n)}}
\right)
\\
&&
\stackrel{d}{\to}
\left\{
\begin{array}{ll}
\left(
E_J, 1_{[0,1]}(t) dt 
\right) & (\alpha > 1/2) \\
\left(
E_J, 
\frac {
\exp
\Bigl(
2 {\cal Z}_{\tau(E_J)\log \frac tU} - 2 \tau(E_J) 
\left|
\log \frac tU 
\right|
\Bigr)
dt
}
{
\int_0^1
\exp
\Bigl(
2 {\cal Z}_{\tau(E_J)\log \frac sU} - 2 \tau(E_J) 
\left|
\log \frac sU 
\right|
\Bigr)
ds
}
\right)
& (\alpha=1/2) \\
\left(
E_J, \delta_{U}(dt)
\right)
& (\alpha < 1/2) 
\end{array}
\right.
\eeq
\end{theorem}

When 
$\alpha < 1/2$, 
this result is the same as 
(\ref{previousresult}) and reflects the fact that, in the global scaling limit,  eigenfunctions are localized around the localization centers being uniformly distributed, which is typical in Anderson localization
\cite{N1}. 
For 
$\alpha > 1/2$,  
this result corresponds to the fact that  
the generalized eigenfunctions are spread over the entire space and is consistent with the extended nature of the system.  
For 
$\alpha = 1/2$, 
since
\beq
\exp
\left[
Z_{\tau(E) \log \frac st} - 
\tau(E) \left| \log \frac st \right|
\right]
=
\exp
\left[
Z_{\tau(E) \log \frac st}
\right]
\begin{cases}
\left( \frac ts \right)^{\tau(E)} & (t<s) \\
\left( \frac st \right)^{\tau(E)} & (s<t)
\end{cases}
\eeq
this result implies that the center 
$U$ 
of the generalized eigenfunction 
$\psi$
is uniformly distributed and 
$\psi$
has the power law decay around 
$U$ 
with Brownian fluctuation. 
Since 
$\lim_{E \downarrow 0} \tau(E) = \infty$ 
and 
$\lim_{E \uparrow \infty} \tau (E) = 0$, 
$\psi$
is localized (resp. delocalized) as  
$E \downarrow 0$
(resp. $E \uparrow \infty$) 
which is consistent with the previous discussion(\ref{Sinebetalimit}).\\

{\bf Remark}\\
{\it 
(1)
The results
in this paper are announced in  
\cite {N4} 
without proof. \\
(2)
In Appendix, 
we discuss the continuum 
1-dimensional operator with decaying coupling constant, and have similar results with Theorem 1.1 except that, 
`` $\log \dfrac tU$" 
for 
$\alpha = 1/2$ 
in Theorem 1.1
is replaced by 
$ t - U $. 
The conclusion for 
$\alpha = 1/2$ 
for this model is essentially the same as that in 
Rifkind-Virag \cite{RV}.}
\\

~For the outline of proof, 
we mostly follow the strategy in \cite{RV}, that is,  
(i) first consider the local version of the problem and then (ii) average over the reference energy, which we show more explicitly below. \\ 
{\bf Step 1 : Local version }\\
We first 
consider the local version 
$\Xi^{(n)}$
of our problem : 
take 
$E_0 > 0$ 
as the reference energy and let 
\beq
\Xi_{E_0}^{(n)}
:=
\sum_j
\delta_{
\Bigl(
n
\bigl(
\sqrt{E_j(n)} - \sqrt{E_0}
\bigr)
+ \theta, 
\,
\mu_{E_j(n)}^{(n)}
\Bigr)
} 
\eeq
which is a point process on 
${\bf R} \times {\cal P}(0,1)$
where 
${\cal P}(0,1)$
is the space of probability measures on 
$(0,1)$ 
with the vague topology.
$\theta$
is a random variable with 
$\theta \sim unif [0, \pi)$
for 
$\alpha > 1/2$ 
which is independent from 
$(X_t)$, 
and 
$\theta = 0$
otherwise. 
The motivation 
to consider 
$\Xi_{E_0}^{(n)}$ 
is to study the behavior of eigenvalues lying in the 
${\cal O}(n^{-1})$-
neighborhood of 
$E_0$
and the measures associated with them. 
Random variable
$\theta$ 
(for 
$\alpha > 1/2$)
has the role of making the 
$n \to \infty$
limit being independent of the choice of subsequence. 
We 
then have the following result which is of independent interest : 
\begin{theorem}
$\Xi^{(n)} \stackrel{d}{\to} \Xi$, 
where
\beq
\Xi
&=&
\left\{
\begin{array}{ll}
\sum_{j \in {\bf Z}}
\delta_{j \pi + \theta}
\otimes
\delta_{ 1_{[0,1]}(t) dt }
& (\alpha > 1/2) \\
\sum_{ \lambda : Sine_{\beta} }
\delta_{
\lambda
}
\otimes
\delta
\Bigl(
\frac { 
\exp ( 2 \widetilde{r}_t(\lambda) ) dt
}
{
\int_0^1
\exp ( 2 \widetilde{r}_s(\lambda) ) ds
}
\Bigr)
& (\alpha = 1/2) \\
\sum_{j \in {\bf Z}}
\delta_{P_j}
\otimes
\delta_{\widetilde{P}_j}, 
\quad
&
(\alpha < 1/2) 
\end{array}
\right.
\eeq
where
$\theta \sim unif [0, \pi)$,  
$\widetilde{r}_t (\lambda)$
is characterized by the following equation : 
\begin{equation}
d 
\widetilde{r}_t (\lambda)
=
\frac {
\tau(E_0)
}{t}
dt
+
\sqrt{
\frac {
\tau(E_0)
}{t}
}
d B_t^{\lambda}, 
\quad
t > 0, 
\quad
\lambda \in {\bf R}, 
\label{SDE}
\end{equation}
and 
$\{ B_t^{\lambda} \}_{\lambda}$
is a family of Brownian motion.
Moreover 
$\{ P_j \} \sim Poisson (d \lambda /\pi)$
and  
$\{ \widetilde{P}_j \} \sim Poisson (1_{[0,1]}(t) dt)$.
%
The intensity measure of 
$\Xi$ 
is given by 
\beq
&&
{\bf E} 
\left[
\int
G(\lambda, \mu) 
d \Xi(\lambda, \mu)
\right]
\\
&=&
\frac {1}{\pi}
\left\{
\begin{array}{ll}
\int d \lambda \,
{\bf E}
\left[
G
\left(
\lambda, 
1_{[0,1]} (t) dt 
\right)
\right]
& 
(\alpha > 1/2) \\
\int d \lambda \,
{\bf E}
\left[
G
\left(
\lambda, 
\frac {
\exp
\Bigl(
2 {\cal Z}_{
\tau(E_0) \log \frac tU
}
-
2 \tau(E_0)
\log 
\left| \frac tU \right|
\Bigr)
dt
}
{
\int_0^1
\exp
\Bigl(
2 {\cal Z}_{
\tau(E_0) \log \frac sU
}
-
2 \tau(E_0)
\log 
\left| \frac sU \right|
\Bigr)
ds
}
\right)
\right]
&
(\alpha = 1/2)
\\
\int d \lambda \,
{\bf E}
\left[
G
\left(
\lambda, 
\delta_{ U }
\right)
\right]
& 
(\alpha < 1/2) 
\end{array}
\right.
\eeq
for 
$G \in C_b({\bf R} \times {\cal P}(0,1))$, 
where 
$U := unif [0,1]$.
\end{theorem}
%
We note that
eq.(\ref{SDE})
determines
$\widetilde{r}_t(\lambda)$
up to constant and hence by normalization it  
determines 
$
\exp [ 2\widetilde{r}_t(\lambda) ] dt 
/
\int_0^1
\exp [ 2\widetilde{r}_s(\lambda) ] ds
$
uniquely.
We also note that 
the intensity measure of 
$\Xi_{E}$
gives the limit of the measure part 
$\mu_{E}^{(n)}$
in Theorem 1.1.
The main technical problems 
we have for the proof are : 
(1)
For the critical case,
the radial component of the generalized eigenfunction of 
$H$
are divergent so that the argument in 
\cite{RV},
which works for decaying coupling constant model (DC model, in short), 
is not directly applicable.
Thus 
we need to renormalize it by cancelling out with the term coming from the normalization part. 
(2)
For the sub-critical case, 
we obtained that the joint limit of the pair of rescaling eigenvalues and corresponding localization centers converge to a Poisson process on 
${\bf R} \times [0,1]$, 
which has been proved for a class of random Schr\"odinger operators \cite{KiN, N1, GK}. 
However, 
in \cite{KiN, N1, GK} 
one used the stationarity of the random potential and
Minami's estimate which is known to hold mostly for discrete models only. 
Here 
we have a general argument such that if (i) the reference energy lies in the localized regime and if (ii) the point process of rescaled eigenvalues of any subsystems, of which the volume is comparable with that of the original one, converge to a Poisson process on 
${\bf R}$, 
then the joint limit of eigenvalues and localization centers also converge to a Poisson process on 
${\bf R} \times [0,1]$.\\
{\bf Step 2 : Average over the reference energy }\\
Next, 
as is done in Rifkind-Virag \cite {RV}, 
we take the ``average" the result of Theorem 1.2 over the reference energy
$E_0$ 
w.r.t. the density of states measure 
$dN$, 
leading to the conclusion of Theorem 1.1. 
This is 
a model-independent, general argument and if 
(i)
one has the limit of the local version 
$\Xi^{(n)}_{E_0}$
and its intensity measure, 
and if 
(ii) 
$G_n (E)$
(defined in Section 3)
is uniformly integrable w.r.t. 
$dN \times {\bf P}$, 
then the distribution of the limit of the measure part 
$\mu^{(n)}_E$
is given by the intensity measure.\\
~The rest 
of this paper is organized as follows. 
In sections 2, 3, 
we prove theorems 1.2, 1.1 respectively. 
In Appendix, 
we state the results for DC model and prove uniform integrability mentioned above. 
We 
sometimes use momentum variable 
$\kappa = \sqrt{E}$
instead of energy variable
$E$.
%
%
\section{Local version}
In this section 
we prove Theorem 1.2 separately for 
$\alpha > 1/2$, 
$\alpha = 1/2$, 
and 
$\alpha < 1/2$.
%
\setcounter{subsection}{-1}
\subsection{Preliminary}
We adopt 
the following version of Pr\"ufer coordinate for the solution 
$\psi_E$
to the Schr\"odinger equation 
$H \psi_{\kappa^2}(t) = \kappa^2 \psi_{\kappa^2}(t)$, 
$\kappa > 0$ :  
\begin{equation}
\left(
\begin{array}{c}
\psi_{\kappa^2}(t) \\
\psi'_{\kappa^2}(t)/\kappa
\end{array}
\right)
:=
R_t (\kappa)
\left(
\begin{array}{c}
\sin \theta_t(\kappa) \\
\cos \theta_t(\kappa)
\end{array}
\right). 
\label{Prufer}
\end{equation}
Introducing 
$\widetilde{\theta}_t (\kappa)$, 
$r_t (\kappa)$ 
by 
\beq
\theta_t(\kappa)
=:
\kappa t + \widetilde{\theta}_t (\kappa), 
\quad
R_t (\kappa)
=:
\exp
\left[
r_t(\kappa)
\right]
\eeq
we have 
\cite{KU}
\begin{eqnarray}
r_t (\kappa)
&=&
\frac {1}{2 \kappa}
Im
\,
\int_0^t 
a(s) e^{2 i \theta_s (\kappa) }
F(X_s) ds, 
\label{Prufer1}
\\
\widetilde{\theta}_t (\kappa)
& := &
\frac {1}{2 \kappa}
Re \,
\int_0^t 
\left(
e^{2i \theta_s(\kappa)} - 1 
\right)
a(s)  F(X_s) ds.
\label{Prufer2}
\end{eqnarray}
%
\subsection{Super-critical case}
{\it Proof of Theorem 1.2 for super-critical case}\\
In general, 
for a 
sequence 
$\{ \xi_n \}$
of the point processes 
on the metric space 
$X$ 
and a point process 
$\xi$, 
$\xi_n \stackrel{d}{\to} \xi$
is equivalent to 
$\int h d \xi_n 
\stackrel{d}{\to}
\int h d \xi$
for any 
$h \in C_c (X)$
(the space 
${\cal P}(X)$
of probability measures on 
$X$
is endowed with the vague topology).
It thus suffices to show 
\beq
&&
\int h d \Xi_{E_0}^{(n)}
\stackrel{d}{\to}
\int h d \Xi_{E_0}, 
\quad
\Xi_{E_0}
:=
\sum_{j \in {\bf Z}}
\delta_{ j\pi + \theta} 
\otimes
\delta_{ 1_{[0,1]}(t)dt }, 
\quad
\theta \sim unif [0, \pi) 
\eeq
for any 
$h \in C_c({\bf R} \times {\cal P}[0, 1])$.
Also, 
it is sufficient to assume that  
$h (\lambda, \mu)
=
h_1 (\lambda) \cdot h_2 (\mu)$ 
with 
$h_1 \in C_c ({\bf R})$, 
$h_2 \in C({\cal P}(0,1))$.
We 
first work under a subsequence 
$\{ n_k \}_k$
which satisfies the following condition : 
$\lim_{k \to \infty} n_k = \infty$
and there exist
$\{ m_k \}_k \subset {\bf N}$
and 
$\beta \in [0, \pi)$ 
such that 
\begin{equation}
n_k \sqrt{E_0} = m_k \pi + \beta + o(1), 
\quad
k \to \infty. 
\label{subsequence}
\end{equation}
Here 
we make use of the facts that, for a.s.,  
$\widetilde{\theta}_t(\kappa)
\stackrel{t \to \infty}{\to}
\widetilde{\theta}_{\infty}(\kappa)$
for locally uniformly w.r.t. 
$\kappa$
\cite{KU},
and 
$E_j (n) \to E_0$
for all 
$j$'s 
such that 
$n \left(
\sqrt{E_j(n)} - \sqrt{E_0}
\right) + \theta \in \mbox{supp } h_1$
(\cite{KN1}, Lemma 4.1).
By Sturm's 
oscillation theory, we have 
\beq
j \pi 
= \theta_{n_k} 
\left(
\sqrt{E_j(n_k)}
\right)
&=&
\sqrt{E_j(n_k)} n_k + 
\widetilde{\theta}_{n_k} 
\left(
\sqrt{E_j(n_k)}
\right)
\\
&=& 
\sqrt{E_j(n_k)} n_k
+ 
\widetilde{\theta}_{\infty}
\left(
\sqrt{ E_0 }
\right)
+
o(1), 
\quad
k \to \infty
\eeq
and thus together with 
(\ref{subsequence}), 
\beq
n_k \left( \sqrt{E_j(n_k)} - \sqrt{E_0} \right) + \theta
& \stackrel{a.s.}{=} &
(j - m_k) \pi - \widetilde{\theta}_{\infty} 
\left(
\sqrt{ E_0 }
\right)
 - \beta + \theta + o(1).
\eeq
which yields 
\begin{equation}
\sum_j h_1 
\left(
n_k \left( \sqrt{E_j(n_k)} - \sqrt{E_0} + \theta \right)
\right)
\stackrel{a.s.}{\to}
\sum_j h_1 
\left(
(j - m_k) \pi - \widetilde{\theta}_{\infty} 
\left(
\sqrt{ E_0 }
\right)
- \beta+ \theta
\right).
\label{h1}
\end{equation}
To study 
the measure part, we introduce measures 
$\nu^{(n)}_E, \mu^{(n)}_E$
on 
$(0,1)$ : 
\beq
d\nu_{E}^{(n)}=
n  \cdot R_{nt}(\sqrt{E})^2 dt, 
\quad
\mu_{E}^{(n)}
:=
\frac {
\nu_E^{(n)}
}
{
\nu_E^{(n)} (0,1)
}
\eeq
so that 
$\mu_E^{(n)}$
is equal to the measure part 
$\mu_{E_j(n)}^{(n)}$
of 
$\Xi_{E_0}^{(n)}$
if 
$E$
is equal to an eigenvalue 
$E_j(n)$ 
of 
$H_n$.
Then by Lemma \ref{locunif}
given below, 
for a.s. and for locally uniformly 
w.r.t. 
$E$, 
we have 
\begin{equation}
h_2 (\mu_{E}^{(n)}) 
\stackrel{a.s.}{\to} 
h_2(1_{[0,1]}(t)dt). 
\quad
\label{h2}
\end{equation}
Therefore 
by $(\ref{h1}), (\ref{h2})$, 
\beq
\sum_j h_1 
\left(
n_k \left( \sqrt{E_j} - \sqrt{E_0} \right)
+ \theta
\right)
h_2 (\mu_{E_j (n_k)}^{(n_k)})
\stackrel{a.s.}{\to}
\sum_{j \in {\bf Z}} 
h_1 
\left(
j\pi - \widetilde{\theta}_{\infty} 
\left( \sqrt{E_0} \right)
- \beta + \theta
\right)
h_2(1_{[0,1]}(t) dt).
\eeq
Here we note that 
$(\beta- \widetilde{\theta}_{\infty} (\kappa) + \theta)_{\pi {\bf Z}} \stackrel{d}{=} \theta$, 
where 
$(x)_{\pi {\bf Z}}
:=
x - 
\max \{ k \in {\bf Z} \, | \, k \le x \}
\in
[0, \pi)$
is the ``fractional part" of 
$x$ 
modulo 
$\pi {\bf Z}$, 
which yields 
\beq
&&
\int h\, d \,\Xi_{E_0}^{(n_k)}
\stackrel{d}{\to}
\int h\, d \,\Xi_{E_0}, 
\quad
h \in C_c ({\bf R} \times {\cal P}(0,1)).
\eeq
Since 
the limit in distribution is independent of the subsequence, this convergence holds for the whole limit so that we have  
$\Xi_{E_0}^{(n)} \stackrel{d}{\to} \Xi_{E_0}$. 
For the intensity measure, 
we note 
$\theta \sim unif [0, \pi)$ 
and compute : 
\beq
{\bf E} 
\left[
\int G(\lambda, \mu)
\,
d \Xi_{E_0}
\right]
=
\sum_{j \in {\bf Z}}
\int_0^{\pi}
\frac {d \theta}{\pi}
G
\left(
j \pi + \theta, 1_{[0,1]}(t) dt
\right)
=
\frac {1}{\pi}
\int d \lambda\,
G
\left(
\lambda, 1_{[0,1]}(t) dt
\right). 
\eeq
\QED
\\

\begin{lemma}
\label{locunif}
For a.s., we have
\beq
d\mu_{E}^{(n)}
\stackrel{v}{\to} 1_{[0,1]}(t) dt. 
\eeq
locally uniformly w.r.t. 
$E$.
\end{lemma}
%
\begin{proof}
Let 
$\kappa := \sqrt{E}$.
Since 
$r_t(\kappa)
\stackrel{t \to \infty}{\to}
r_{\infty}(\kappa)$
locally uniformly w.r.t.
$\kappa$ 
for almost surely 
\cite{KU}, 
for any 
$\epsilon > 0$ 
there exists 
$T_{\epsilon} > 0$ 
s.t. 
$t > T_{\epsilon}$
implies 
$|
r_t (\kappa)
-
r_{\infty} (\kappa)
| < \epsilon$.
By definition, 
%
$
\nu_{E}^{(n)} (a,b)
=
\int_a^b 
e^{2 r_{nu}(\kappa)} n du
=
\int_{na}^{nb} 
e^{ 2 r_s(\kappa) } ds
$ 
%
so that for 
$na > T_{\epsilon}$, 
\begin{equation}
e^{2 r_{\infty}(\kappa) - 2 \epsilon}
n(a-b)
\le
\int_{na}^{nb} 
e^{ 2 r_s (\kappa) } ds
\le
e^{2 r_{\infty}(\kappa) + 2 \epsilon}
n(a-b).
\label{one}
\end{equation}
To estimate  
$
\nu_{E}^{(n)}(0,1)
=
\int_0^1 
e^{2 r_{nu} } n du
=
\int_{0}^{n} 
e^{ 2 r_s } ds
$, 
we note that 
$\sup_{n, t} | r_{nt}(\kappa) | =: M < \infty$, 
and we divide the domain of integral as 
$(0,n) = (0, T_{\epsilon} ] \cup [ T_{\epsilon}, n)$. 
We then have
\begin{equation}
e^{2 r_{\infty}(\kappa) - 2 \epsilon}
(n - T_{\epsilon})
+
T_{\epsilon} e^{- 2M}
\le
\int_{0}^{n} 
e^{ 2 r_s (\kappa)} ds
\le
e^{2 r_{\infty}(\kappa) + 2 \epsilon}
(n - T_{\epsilon})
+
T_{\epsilon} e^{ 2M}.
\label{two}
\end{equation}
By
(\ref{one}), (\ref{two})
\beq
e^{- 4 \epsilon}
(a-b)
\le
\liminf_{n \to \infty}
\frac {
\nu_{E}^{(n)} (a,b)
}
{
\nu_{E}^{(n)} (0,1)
}
\le
\limsup_{n \to \infty}
\frac {
\nu_{E}^{(n)} (a,b)
}
{
\nu_{E}^{(n)} (0,1)
}
\le
e^{ 4 \epsilon}
(a-b).
\eeq
Since 
$\epsilon > 0$
is arbitrary, we have 
\beq
\lim_{n \to \infty}
\frac {
\nu_{E}^{(n)} (a,b)
}
{
\nu_{E}^{(n)} (0,1)
}
=
(a-b), 
\quad
a.s.
\eeq
\QED
\end{proof}
%

\subsection{Critical case}
In section 2.2, 
we renormalize the radial component of 
$\psi_{ \kappa^2 }$
in section 2.2.1, 
and 
prove Theorem 1.2 for 
$\alpha = 1/2$ 
in section 2.2.2.
\subsubsection{Renormalize the radial component}
Let 
$r_t(\kappa)$, 
$\widetilde{\theta}_t (\kappa)$
defined in 
(\ref{Prufer1}), (\ref{Prufer2}).
And let 
$\widetilde{r}^{(n)}_t(\kappa)
:=
r_{nt}(\kappa)
-
\langle F g_{\kappa} \rangle
\int_0^{n} a(s)^2 ds$
be the ``renormalized" radial part of 
$\psi_{\kappa^2}(t)$, 
where 
$g_{\kappa} := (L + 2i \kappa)^{-1} F$, 
$t \in (0,1)$.
To study 
the local version of our problem, we work under the following notation :  
$\kappa_{c} := \kappa_0 + \frac {c}{n}$, 
$\kappa_0 := \sqrt{E_0}$, 
$c \in {\bf R}$. 
We then have
%

\begin{lemma}
\label{Radial}
If 
$\alpha = 1/2$, 
then there exists subsequence 
$\{ n_k \}_{k \ge 1}$
and continuous function-valued process
$\widetilde{r}_t (c)$
such that 
\beq
&&
\widetilde{r}^{(n_k)}_t (\kappa_{c})
\stackrel{d}{\to}
\widetilde{r}_t(c), 
\quad
\mbox{ locally uniformly in }
t \in (0,1),
\, 
c \in {\bf R} 
\\
&&
d 
\widetilde{r}_t (c)
=
\frac {
\tau(\kappa_0^2)
}{t}
dt
+
\sqrt{
\frac {
\tau(\kappa_0^2)
}{t}
}
d B_t^{c}, 
\quad
t > 0
\eeq
where 
$\{ B_t^{c} \}$
is a family of Brownian motion.
\end{lemma}
\begin{proof}
Letting 
\beq
J_t (\kappa)
&:=&
\int_0^t 
a(s) e^{2i \theta_s(\kappa) } F(X_s) ds, 
\quad
J_t(0)
:=
\int_0^t 
a(s) 
F(X_s) ds
\eeq
we have
\beq
r_t(\kappa)
&=&
\frac {1}{2 \kappa} Im \, J_t (\kappa), 
\quad
\widetilde{\theta}_t(\kappa)
=
\frac {1}{2 \kappa} 
Re 
\left(
J_t (\kappa) - J_t(0)
\right).
\eeq
Here 
we use the following lemma in 
\cite{KN1}. \\
%

%
{\bf Lemma 6.2 in \cite{KN1}}\\
(1)
\beq
\int_0^t
a(s) e^{2i \theta_s(\kappa)} F(X_s) ds
&=&
- \frac {i}{2 \kappa}
\int_0^t a(s)^2 F g_{\kappa} (X_s) ds
+
Y_t (\kappa)
+
\delta_t(\kappa)
\\
\mbox{ where }
\quad
Y_t (\kappa)
&:=&
\int_0^t 
a(s) e^{2i \theta_s(\kappa)}
\nabla g_{\kappa} (X_s) d X_s
\\
\delta_t (\kappa)
&:=&
\left[
a(s) e^{2i \theta_s(\kappa)} g_{\kappa} (X_s)
\right]_0^t
-
\int_0^t a'(s) e^{2i \theta_s(\kappa)} g_{\kappa}(X_s) ds
\\
&&
-
\frac {i}{\kappa}
\int_0^t
a(s)^2
\left(
\frac {e^{2i \theta_s(\kappa)}}{2} - 1
\right)
e^{2i \theta_s(\kappa)}
F g_{\kappa} (X_s) ds, 
\eeq
(2)
$\lim_{t \to \infty} \delta_t (\kappa)
=
\delta_{\infty}(\kappa)$, 
a.s. for some 
$\delta_{\infty}(\kappa)$,
\\
(3)
$\lim_{n \to \infty}
{\bf E}
\left[
\max_{0 \le t \le T}
\left|
\delta_{nt}(\kappa_c) - \delta_{nt} (\kappa_0)
\right|^2 
\right]
= 0$.\\
%
%

By Ito's formula, 
\begin{equation}
Fg_{\kappa} (X_s) ds
=
\langle F g_{\kappa} \rangle ds
+
d L^{-1}
\left(
F g_{\kappa} - \langle F g_{\kappa} \rangle
\right)
-
\nabla L^{-1} 
\left(
F g_{\kappa} - \langle F g_{\kappa} \rangle
\right)
d X_s
\label{Ito}
\end{equation}
by which we further integrate by parts. 
\beq
\int_0^t a(s)^2 F g_{\kappa} (X_s) ds
&=&
\langle F g_{\kappa} \rangle
\int_0^t a(s)^2 ds
+
\widetilde{\delta}_t (\kappa)
\\
\mbox{ where }
\quad
\widetilde{\delta}_t (\kappa)
&:=&
\left[
a(s)^2 L^{-1} 
\left(
F g_{\kappa} - \langle F g_{\kappa} \rangle
\right)
\right]_0^t 
\\
&&
-
\int_0^t (a(s)^2)' L^{-1}
\left(
F g_{\kappa} - \langle F g_{\kappa} \rangle
\right)
ds
-
\int_0^t a(s)^2 \nabla L^{-1}
\left(
F g_{\kappa} - \langle F g_{\kappa} \rangle
\right)
d X_s. 
\eeq
$\widetilde{\delta}_t(\kappa)$
is a sum of convergent terms and one with finite quadratic variation so that it converges almost surely : 
\beq
\lim_{t \to \infty}
\widetilde{\delta}_t (\kappa)
=
\widetilde{\delta}_{\infty} (\kappa), 
\quad
a.s.
\eeq
Here 
we replace 
$t$
by 
$nt$, 
and take 
$n \to \infty$
limit with 
$t$ 
being fixed.
To cancel out 
the divergent term, we subtract 
%
$
\langle F g_{\kappa} \rangle
\int_0^n a(s)^2 ds
$
%
from both sides and obtain 
\beq
\int_0^{nt} a(s)^2 F g_{\kappa} (X_s) ds
-
\langle F g_{\kappa} \rangle
\int_0^n a(s)^2 ds
&=&
\langle F g_{\kappa} \rangle
\log t  
+
\widetilde{\delta}_{\infty} (\kappa)
+ \epsilon_n (t), 
\quad
\lim_{n \to \infty} \epsilon_n (t) = 0.
\eeq
Therefore
\beq
J_{nt}(\kappa)
&=&
- \frac {i}{2 \kappa}
\langle F g_{\kappa} \rangle
\log t
+
Y_{nt} (\kappa)
+ A_n + \epsilon'_n (t), 
\quad
\lim_{n \to \infty} \epsilon'_n (t) = 0.
\\
where 
\quad
A_n 
&:=&
- \frac {i}{2 \kappa}
\left(
\langle F g_{\kappa} \rangle
\int_0^n a(s)^2 ds
+
\widetilde{\delta}_{\infty} (\kappa)
\right)
+
\delta_{\infty}(\kappa)
\eeq
We remark that 
$A_n$
is random but does not depend on 
$t$.
Let
\beq
r_t^{(n)} (\kappa) := r_{nt} (\kappa).  
\quad
\eeq
Then we have
\beq
r_{t}^{(n)}(\kappa)
&=&
\frac {1}{2 \kappa}
Im \;
[J_{nt}(\kappa)]
=
\widetilde{r}_t^{(n)} (\kappa)
+
\widetilde{A}_n
+
\widetilde{\epsilon}_n(t), 
\\
where 
\quad
\widetilde{r}^{(n)}_t (\kappa)
&:=&
\frac {1}{2 \kappa}
Im \,
\Bigl(
- \frac {i}{2 \kappa}
\langle F g_{\kappa} \rangle \log t 
+
Y_{nt} (\kappa)
\Bigr)
\\
and
\quad
\widetilde{A}_n
&:=&
\frac {1}{2 \kappa}
Im \;[ A_n ], 
\quad
\lim_{n \to \infty}
\sup_{
\frac {\sqrt{\log n}}{n} \le t
}
\widetilde{\epsilon}_n (t) = 0.
\eeq
The density function of 
$\mu_E^{(n)}$
is equal to 
\beq
&&
\frac {
\exp [ 2 r_t^{(n)} ]
}
{
\int_0^1 
\exp [ 2 r_s^{(n)} ]
ds
}
\\
&=&
\frac {
\exp [ 2\widetilde{r}_t^{(n)}]
}
{
\int_0^{ \frac {\sqrt{\log n}}{n} }
\exp \bigl[ 2r_s^{(n)}- 2 \widetilde{A}_n -2\widetilde{\epsilon}_n(t)  
\bigr]
ds
+
\int_{ \frac {\sqrt{\log n}}{n} }^1
\exp [ 2\widetilde{r}_s^{(n)}-2\widetilde{\epsilon}_n(t)+2\widetilde{\epsilon}_n(s)]
ds
}. 
\eeq
To estimate the 1st term 
in the denominator, we note that 
\beq
| r_s^{(n)} |
& \le &
\frac {1}{ 2 \kappa }
\left|
Im \, 
\int_0^{ns}
e^{2i \theta_u(\kappa)} a(u) F(X_u) du
\right|
\le 
(Const.) 
\sqrt{ \log n }, 
\quad
0 \le s \le 
\frac {\sqrt{\log n}}{n}. 
\eeq
Moreover, 
$A_n = O(\log n)$, 
and denoting by 
$\sigma_F$
the spectral measure of 
$L$
w.r.t. 
$F$, 
we have 
\beq
\widetilde{A_n}
=
Im 
\Bigl[ - \dfrac {i}{2 \kappa} \langle F g_{\kappa} \rangle
\Bigr]
&=&
\frac {1}{2 \kappa}
\int_{-\infty}^0
\frac { - \lambda}{ \lambda^2 + (2 \kappa)^2 }
d \sigma_F (\lambda)
> 0, 
\eeq
which yields
\beq
\int_0^{ \frac {\sqrt{\log n}}{n} }
\exp \bigl[ r_s^{(n)}-  \widetilde{A}_n -\widetilde{\epsilon}_n(t) 
\bigr]
ds
\le
\frac { \sqrt{ \log n } }{n}
\cdot
\exp \left[
- (Const.) \log n
\right]
=
O(n^{ -\delta }), 
\eeq
for some 
$\delta > 0$. 
On the other hand, 
\beq
\int_0^{ \frac {\sqrt{\log n}}{n} }
\exp 
\left[
\widetilde{r}_s^{(n)} 
\right] ds
=
\int_0^{ \frac {\sqrt{\log n}}{n} }
\exp 
\left[
r_s^{(n)} - \widetilde{A}_n - \widetilde{\epsilon}_n(s)\right] ds
=
O(n^{ -\delta' })
\eeq
Thus we have 
\beq
\frac {
e^{2 r_{t}^{(n)}(\kappa)} dt
}
{
\int_0^1
e^{2 r_{s}^{(n)}(\kappa)} ds
}
-
\frac {
e^{2 \widetilde{r}^{(n)}_{t}(\kappa)} dt
}
{
\int_0^1
e^{2 \widetilde{r}^{(n)}_{s}(\kappa)} ds
}
\stackrel{a.s.}{=} o(1)
\eeq
so that we henceforth consider 
$\widetilde{r}_t^{(n)} (\kappa)$
instead of 
$r_t^{(n)}(\kappa)$.
Let 
\beq
\Theta^{(n)}_t (c)
&:=&
\theta_{nt} (\kappa_c) - \theta_{nt}(\kappa_0), 
\quad
Y^{(n)}_t (\kappa_c)
:=
Y_{nt} (\kappa_c).
\eeq
Then 
\beq
Y^{(n)}_t (\kappa_c)
&=&
\int_0^{nt}
a(s) e^{2i \theta_s(\kappa_c)} 
\nabla g_{\kappa_c} (X_s) d X_s
\\
&=&
\int_0^{nt}
a(s)
e^{2i 
\left(
\theta_s(\kappa_c) - \theta_s(\kappa_0)
\right)
}
e^{2i \theta_s(\kappa_0)}
\nabla g_{\kappa_c} (X_s) dX_s
\quad
\eeq
and by 
\cite{KN1}
Proposition 9.1 and Lemma 9.3, 
we have 
$\Theta_u^{(n)}(c)
\stackrel{d}{\to}
\Theta_u (c)$, 
$\theta_{nu}(\kappa_0) 
\stackrel{d}{\to} 
U$.
Where 
$U \sim unif[0, \pi)$, 
$\Theta_u (c)$ 
and 
$U$
are independent, and these convergence is uniform w.r.t.
$u \in [0,1]$
and 
$c \in {\bf R}$.
Moreover, 
$\Theta_t (c)$
is characterized by the following SDE. 
\beq
d \Theta_t (c)
&=&
c dt
+
\sqrt{ \tau (E_0) }
Re
\left[
\left(
e^{2i \Theta_s (c) } - 1 
\right)
\frac {d Z_t}{\sqrt{t}}
\right], 
\quad
\Theta_0(c) = 0. 
\eeq
$Z_t$
is a complex Brownian motion.
By Skorohod's theorem, 
we can assume that 
$\Theta_u^{(n)}(c)
\stackrel{a.s.}{\to}
\Theta_u (c)$, 
$\theta_{nu}(\kappa_0) 
\stackrel{a.s.}{\to} 
U$.
Thus for 
$0 < s < t$, 
we use 
(\ref{Ito}) 
and integrate by parts to yield
\beq
&&
\left\langle 
Y^{(n)} (\kappa_{c_1}), 
\overline{ Y^{(n)} (\kappa_{c_2}) }
\right\rangle_t
-
\left\langle 
Y^{(n)} (\kappa_{c_1}), 
\overline{
Y^{(n)} (\kappa_{c_2})
}
\right\rangle_s
\\
&=&
\int_s^t 
a(nu)^2
\exp \left[2i 
\left(
\Theta^{(n)}_u (c_1) - \Theta^{(n)}_u (c_2)
\right)
\right]
\left[
g_{\kappa_{c_1}}, g_{\kappa_{c_2}} 
\right]
(X_{nu}) n du
\\
& \stackrel{n \to \infty}{\to} &
\langle 
\left[
g_{\kappa_{c_1}}, g_{\kappa_{c_2}} 
\right]
\rangle
\int_s^t \frac {du}{u} 
\exp \left[2i 
\left(
\Theta_u (c_1) - \Theta_u (c_2)
\right)
\right]
\eeq
and together with martingale inequality, we obtain
\beq
{\bf E}
\left[
| Y_t^{(n)}(\kappa_c) - Y_s^{(n)}(\kappa_c) |^4
\right]
\le
C (t-s)^2, 
\quad
0 < s < t.
\eeq
Therefore, 
for any fixed 
$\epsilon > 0$, 
and for any 
$t, s \in [\epsilon, 1]$, 
we have a tightness condition. 
\beq
&(1)&
\quad
\limsup_{A \to \infty}
\sup_{n > 0}
{\bf P}
\left(
| \widetilde{r}_{nt}(\kappa) | \ge A
\right)
= 0
\\
&(2)&
\quad
\lim_{\delta \downarrow 0}
\limsup_{n \to \infty}
{\bf P}
\left(
\sup_{t, s \in [\epsilon, 1], 
\;
|t-s| < \delta}
| \widetilde{r}^{(n)}_t (\kappa) - \widetilde{r}^{(n)}_s (\kappa) |
> \rho
\right) = 0, 
\quad
\rho > 0.
\eeq
Let 
$\widetilde{r}_t(c)$
be a limit point of 
$\widetilde{r}_t^{(n)} (\kappa_c)$.
We then have 
\beq
d \widetilde{r}_t (c)
=
\frac {
e(\kappa_0)
}
{t}
dt
+
\frac { f(\kappa_0) }{ \sqrt{t} }
Im \, 
\left[
e^{2i \Theta_t (c)} 
d Z_t
\right], 
\quad 
t > 0 
\eeq
where 
\beq
e(\kappa) &:=& 
- 
\frac {1}{2 \kappa}
Im 
\left[
\frac {2i}{2 \kappa}
\cdot 
\frac 12
\langle F g_{\kappa} \rangle
\right]
=
f(\kappa)^2
=
\tau(\kappa^2)
\\
f(\kappa) &:=& 
\frac {1}{2 \kappa}
\sqrt{ \frac {[g_{\kappa}, \overline{g_{\kappa}}] }{2} }
=
\sqrt{\tau(\kappa^2)}
\eeq
Letting 
$B_t^c := Im [ e^{2i \Theta_t (c)} Z_t ]$, 
we complete the proof of Lemma \ref{Radial}. 
\QED
\end{proof}
%
We next consider 
%
$
\phi_t (c)
:=
\dfrac {\partial \Theta_t(c)}{\partial c}
$
which satisfies 
\beq
d \phi_t(c)
&=&
dt - 
\sqrt{
\frac {\tau(\kappa_0^2)}{t}
}
d B_t^c
\cdot
2 \phi_t(c), 
\quad
\phi_0 (c) = 0.
\eeq
Since 
the joint distribution of 
$\widetilde{r}_t (c)$
and 
$\phi_t (c)$
does not depend on 
$c$
so that we ignore the 
$c$-dependence of 
$B_t^c$. 
They now satisfy 
\beq
d 
\widetilde{r}_t (c)
&=&
\frac {
\tau(\kappa_0^2)
}{t}
dt
+
\sqrt{
\frac {
\tau(\kappa_0^2)
}{t}
}
d B_t^{c}, 
\quad
t > 0
\\
d \phi_t(c)
&=&
dt - 
\sqrt{
\frac {\tau(\kappa_0^2)}{t}
}
d B_t^c
\cdot
2 \phi_t(c), 
\quad
\phi_0 (c) = 0.
\eeq
By the 
change of variable, 
\beq
t = t(v) = 
\exp
\left[
\frac {v}{\tau(\kappa_0^2)}
\right], 
\quad
v(t) = \tau(\kappa_0^2) \log t, 
\quad
t \in (0, 1), 
\;
v \in (-\infty, 0)
\eeq
we have
\beq
d\widetilde{r} (c)
&=&
dv + d B_v
\\
d \phi
&=&
\frac {t}{\tau(\kappa_0^2)}dv 
-
2 \phi \cdot d B_v
\eeq
so that
\begin{equation}
\phi_w
=
\int_{-\infty}^w
\frac {t(v)}{\tau(\kappa_0^2)}
e^{\widetilde{r}_v - \widetilde{r}_w}dv
\label{phi}
\end{equation}
which is the same equation satisfied by 
$r_t (c), \phi_t (c)$
in DC model derived in 
\cite{RV}, except that we have 
$v \in (-\infty, 1]$ 
while 
$t \in [0,1]$
in DC model. 
As is done in 
\cite{RV}, 
by  Girsanov's formula we have

\begin{lemma}
\label{Girsanov}
Let  
$\epsilon > 0$. 
Let 
$R(\omega)$ 
be the distribution of 
$\widetilde{r}_v (c)$
as an element of 
$C[v(\epsilon), v(1)]$.
Then under 
$d Q(\omega) := 
e^{\omega_v - \omega_{v(1)}} d R(\omega)$
we have 
\beq
\widetilde{r}_s
=
f^v (s) + \widehat{B}_s, 
\quad
f^v (s)
:=
v - |s - v|, 
\quad
s \in [v(\epsilon), v(1)]
\eeq
where 
$\widehat{B}_s$
is a Brownian motion. 
\end{lemma}
%
%
%
\subsubsection{Local version for the critical case}
In this section 
we prove Theorem 1.2 for 
$\alpha = 1/2$. 
Let 
$\{ n_k \}$, 
$\widetilde{r}_t (c)$
be those in  
Lemma \ref{Radial}.
%
%
\begin{lemma}
\label{Localcritical}
Let 
\beq
\Lambda_{n, E_0}
:=
\left\{
n \left(
\sqrt{E_j (n)} - \sqrt{E_0}
\right)
\right\}_{j \ge 1}.
\eeq
Then we have 
\beq
&&
\left\{
\left(
\lambda, e^{\widetilde{r}^{(n_k)}_t (\kappa_{\lambda})} 
\right)
\, \middle| \, 
\lambda \in \Lambda_{n_k, E_0}
\right\}
\stackrel{d}{\to}
\left\{
\left( 
\lambda, e^{ \widetilde{r}_t (\lambda) }
\right)
\, \middle| \, 
\lambda \in Sine_{\beta}
\right\}
\\
&&
\left\{
\left(
\lambda, 
\mu_{ \kappa_{\lambda}^2 }^{(n_k)}
\right)
\, \middle| \, 
\lambda
\in \Lambda_{n_k, E_0}
\right\}
\stackrel{d}{\to}
\left\{
\left( 
\lambda, 
\frac {
e^{ 2\widetilde{r}_t (\lambda) }dt
}
{
\int_0^1
e^{ 2\widetilde{r}_s (\lambda) }
ds
}
\right)
\, \middle| \, 
\lambda \in Sine_{\beta}
\right\}
\quad
\eeq
in the sence of convergence in distribution on the sequences of point processes on 
${\bf R} \times C(0,1)$
and 
${\bf R} \times {\cal P}(0,1)$
respectively.
\end{lemma}
%

Any limit point 
$\widetilde{r}_t (c)$
of 
$\widetilde{r}_t^{(n)}(\kappa_c)$
satisfies 
(\ref{SDE})
so that they only differ up to constant.
Hence 
$
e^{ 2\widetilde{r}_t (\lambda) }dt
/
\int_0^1
e^{ 2\widetilde{r}_s (\lambda) }
ds
$
is uniquely determined and we do not need to take subsequence anymore which yields  
%

%
\begin{corollary}
\beq
\left\{
\left(
\lambda, 
\mu_{ \kappa_{\lambda}^2 }^{(n)}
\right)
\, \middle| \, 
\lambda \in \Lambda_{n, E_0}
\right\}
\stackrel{d}{\to}
\left\{
\left( 
\lambda, 
\frac {
e^{ 2\widetilde{r}_t (\lambda) }dt
}
{
\int_0^1
e^{ 2\widetilde{r}_s (\lambda) }
ds
}
\right)
\, \middle| \, 
\lambda \in Sine_{\beta}
\right\}.
\eeq
\end{corollary}
%

%
{\it Proof of Lemma \ref{Localcritical}}\\
As in 
the proof of Theorem 1.2 for super-critical case, we use 
$\xi_n \stackrel{d}{\to} \xi
\Longleftrightarrow
\xi_n (f) \stackrel{d}{\to} \xi (f)$, 
$f \in C_{c}({\bf R} \times C(0,1))$, and assume 
$f(\lambda, \phi) = h_1 (\lambda) \cdot h_2(\phi)$, 
$h_1 \in C_c ({\bf R})$, 
$h_2 \in C(C[0,1])$.
Let 
$\theta^{(n)}_t (\kappa)
:=
\theta_{nt} (\kappa)$
and we recall 
\beq
&&
\lambda \in \Lambda_{n, E_0}
\quad
\Longleftrightarrow
\quad
\Theta^{(n)}_1(\lambda)
\in 
\pi {\bf Z} - ( \theta_1^{(n)}(\kappa_{0}) )_{\pi {\bf Z}}
\\
&&
\lambda \in Sine_{\beta}
\quad
\stackrel{def}{\Longleftrightarrow}
\quad
\Theta_1 (\lambda) 
\in 
\pi {\bf Z} + U
\\
&&
\theta^{(n)}_t (\kappa_{0}) 
\stackrel{d}{\to} 
U
\eeq
where 
$U \sim unif [0, \pi)$
and independent from 
$\Theta_1 (\lambda)$. 
Since 
$\widetilde{r}_t^{(n_k)} (\kappa_{c})$
converges to 
$\widetilde{r}_t(c)$
locally uniformly w.r.t. 
$t \in (0,1)$ 
and 
$c \in {\bf R}$, 
\cite{KN1} Lemma 3.1
yields
\beq
\sum_{\lambda\, : \,
\Theta_1^{(n)}(\lambda) \in \pi {\bf Z} - (\theta^{(n)}_1(\kappa_0))_{ \pi {\bf Z}}}
h_1 (\lambda)
h_2 (\widetilde{r_t}^{(n)}(\kappa_{\lambda}))
& \stackrel{d}{\to} &
\sum_{\lambda\, : \,
\Theta_1(\lambda) \in \pi {\bf Z} - unif [0, \pi)}
h_1 (\lambda)
h_2 (\widetilde{r_t}(\lambda))
\eeq
which implies 
\beq
\left\{
(\lambda, \widetilde{r}_{t}^{(n)}(\kappa_{\lambda}))
\, \middle| \,
\lambda \in \Lambda_{n, E_0}
\right\}
\stackrel{d}{\to}
\left\{
(\lambda, \widetilde{r}_{t}(\lambda))
\, \middle| \,
\lambda \in Sine_{\beta}
\right\}. 
\eeq
We note that 
$\widetilde{r}_t^{(n)}
\stackrel{loc. unif}{\to} 
\widetilde{r}_t$
implies 
$e^{\widetilde{r}_t^{(n)}}
\stackrel{loc. unif}{\to} 
e^{\widetilde{r}_t}$, 
and furthermore 
$e^{\widetilde{r}_t^{(n)}}
\stackrel{loc. unif}{\to} 
e^{\widetilde{r}_t}$
in turn implies 
$e^{\widetilde{r}_t^{(n)}}dt
\stackrel{v}{\to} 
e^{\widetilde{r}_t}dt$. 
By using 
continuous mapping theorem twice, we have 
\beq
e^{\widetilde{r}_t^{(n)}}dt
\stackrel{d}{\to} 
e^{\widetilde{r}_t}dt
\eeq
which proves the first statement. 
The second one is proved similarly.
\QED\\
We turn to  
derive the intensity measure of 
$\Xi_{E_0}$. 
%

%
\begin{lemma}
\label{intensity}
For 
$G \in C_b({\bf R} \times C(0,1))$, 
one has
\beq
{\bf E}
\left[
\sum_{\lambda \in Sine_{\beta} }
G(\lambda, \widetilde{r}(\lambda))
\right]
&=&
\frac {1}{\pi}
\int_{\bf R} d \lambda
\,
{\bf E}
\left[
G(\lambda, 
B_{v(\cdot)} + v(U) - |v(\cdot) - v(U)|
)
\right].
\eeq
where 
\beq
v(t) = \tau (\kappa_0^2)\log t, 
\quad
U \sim unif [0,1].
\eeq
\end{lemma}
%
%
\begin{proof}
The proof 
is almost parallel as that given in 
\cite{RV}. 
In fact, 
\begin{eqnarray}
{\bf E}
\left[
\sum_{\lambda \in Sine_{\beta}}
G(\lambda, \widetilde{r}(\lambda))
\right]
&=&
{\bf E}\left[
\int_0^{\pi}
\frac {du}{\pi}
\sum_{\lambda : 
\Theta_1(\lambda) \in \pi {\bf Z} + u}
G(\lambda, \widetilde{r}(\lambda))
\right]
\nonumber
\\
&=&
{\bf E}
\left[
\frac {1}{\pi}
\int_{\bf R} du 
\sum_{\lambda \,: \,\Theta_1(\lambda) =   u}
G \left(
\lambda, \widetilde{r} (\lambda) 
\right)
\right]
\nonumber
\\
&=&
\frac {1}{\pi}
\int_{\bf R} d \lambda
{\bf E}
\left[
G(\lambda, \widetilde{r} (\lambda))
\left|
\frac {\partial \Theta_1(\lambda)}{\partial \lambda}
\right|
\right]
\label{conclusion}
\end{eqnarray}
where we set 
$u = \dfrac {\partial \Theta_1(\lambda)}{\partial \lambda} d \lambda$.
By
(\ref{phi}) and Lemma \ref{Girsanov}
\beq
{\bf E}
\left[
G(\lambda, \widetilde{r} (\lambda))
\left|
\frac {\partial \Theta_1(\lambda)}{\partial \lambda}
\right|
\right]
&=&
\lim_{\epsilon \downarrow 0}
\int_{v(\epsilon)}^{v(1)}
\frac {t(v)}{\tau(E_0)}
{\bf E}
\left[
G(\lambda, \widetilde{r} (\lambda))
\exp 
\left(
\widetilde{r}_v - \widetilde{r}_{v(1)}
\right)
\right]
dv
\\
&=&
\lim_{\epsilon \downarrow 0}
\int_{v(\epsilon)}^{v(1)}
\frac {t(v)}{\tau(E_0)}
{\bf E}
\left[
G(\lambda, B_{\cdot} + v - |\cdot - v|)
\right]
dv
\\
&=&
\lim_{ \epsilon \downarrow 0 }
\int_{\epsilon}^1
{\bf E}
\left[
G(\lambda, 
B_{v(\cdot)} + v(t) - |v(\cdot) - v(t)|
)
\right] dt
\quad
\\
&=&
\int_{0}^1
{\bf E}
\left[
G(\lambda, 
B_{v(\cdot)} + v(t) - |v(\cdot) - v(t)|
)
\right] dt
\eeq
Substituting this equation into 
(\ref{conclusion})
yields the conclusion. 
\QED
\end{proof}
To drive 
the intensity measure of 
$\Xi_{E_0}$, 
let 
$G \in C_b ({\bf R} \times {\cal P}(0,1))$.
By 
\beq
\int 
G(\lambda, \mu)
d \Xi_{E_0}(\lambda, \mu)
&=&
\sum_{\lambda \in Sine_{\beta}}
G
\left(
\lambda, 
\frac {
e^{2\widetilde{r}_t(\lambda)} dt
}
{
\int_0^1
e^{2\widetilde{r}_s(\lambda)} 
ds
}
\right)
\eeq
and by 
Lemma \ref{intensity}, we have 
\beq
{\bf E}
\left[
\int G(\lambda, \nu) d \Xi_{E_0} (\lambda, \nu)
\right]
&=&
\frac {1}{\pi}
\int_{\bf R} d \lambda
{\bf E}
\left[
G
\left(
\lambda, 
\frac {
\exp 
\left[
2 (B_{v(\cdot)} + v(U) - 2|v(\cdot)-v(U)|)
\right]
}
{
\int_0^1
\exp 
\left[
2 (B_{v(s)} + v(U) - 2|v(s)-v(U)|)
\right]
ds
}
\right)
\right].
\eeq
Here 
we cancel 
$v(U)$
out and use 
$B_{t} 
\stackrel{d}{=} 
Z_{t-v(U)}$ 
+ (random constant)
yielding 
\beq
&=&
\frac {1}{\pi}
\int_{\bf R} d \lambda
{\bf E}
\left[
G
\left(
\lambda, 
\frac {
\exp 
\left[
2 (Z_{v(\cdot)-v(U)}  - 2|v(\cdot)-v(U)|)
\right]
}
{
\int_0^1
\exp 
\left[
2 (Z_{v(s)-v(U)}  - 2|v(s)-v(U)|)
\right]
ds
}
\right)
\right]
\eeq
where 
$Z$ 
is a two-sided Brownian motion.
Together with 
Lemma \ref{Localcritical}
we complete the proof of Theorem 2.1 for 
$\alpha= 1/2$. 
%

\subsection{Sub-critical case}
Let 
$\{ E_j (n) \}$ 
be the positive eigenvalues of 
$H_n$, 
let 
$\psi_{E_j(n)}$
be the normalized eigenfunction corresponding to 
$E_j(n)$
and let 
$x_j (n)$
be a maximal point of 
$| \psi_{E_j(n)}(x) |^2$.
Since 
$\psi_{E_j(n)}$
satisfies the sub-exponential decay estimate : 
$| \psi_{E_j(n)}  (x) |
\le
C 
\exp 
\left[
- D |x - x_j|^{\gamma} 
\right]$, 
$\gamma :=1 - 2 \alpha$, 
maximal points of 
$|\psi_{E_j(n)}|$ 
have a same limit point when they are divided by 
$n$, 
so that we have no ambiguity in choosing 
$x_j (n)$,  
which we call the localization center of 
$E_j (n)$.
Let 
$\xi_n$ 
be the point process of pairs of rescaled eigenvalues and corresponding localization centers, and let 
$\xi$
be a Poisson process. 
\beq
\xi_n
:=
\sum_j 
\delta_{ 
\left(
n( \sqrt{E_j} - \sqrt{E_0}) , 
x_j / n
\right)
}, 
\quad
\xi
= \sum_j \delta_{(P_j, \widetilde{P}_j)}
\sim Poisson (d \lambda/\pi \times 1_{[0,1]}(x)dx). 
\eeq
Theorem 1.2 for 
$\alpha < 1/2$
will follow from the following Proposition.
%

%
%
\begin{proposition}
\label{locctr}
For any bounded intervals
$I (\subset {\bf R})$, 
$B = [a,b]  (\subset (0,1))$, 
we have 
\beq
&(1)& \quad
\lim_n
P( \xi_n (I \times B) = 0)
=
P( \xi (I \times B) = 0)
\\
&(2)& \quad
\lim_n
{\bf E}[ \xi_n (I \times B) ] 
=
{\bf E}[ \xi (I \times B)].
\eeq
Furthermore, 
$\xi_n \stackrel{d}{\to} \xi$.
\end{proposition}
%
%
This 
result was expected to hold true in 
\cite{KN2}. 
For proof, 
we take 
$0 < \delta < 1$
and let 
$C$
(resp. $D$) 
be an interval by eliminating
(resp. adding) a small interval of width 
$n^{\delta}$ 
from 
(resp. to)
$nB :=n[a,b]$. 
\beq
C := 
[an + n^{\delta}, bn - n^{\delta}], 
\quad
D :=
[an - n^{\delta}, bn + n^{\delta}]. 
\eeq
Let 
$H_C := H_n |_C$, 
$H_D := H_n |_D$
with Dirichlet boundary condition, 
and let 
$\xi_n^C$, 
$\xi_n^D$
be the point processes such that 
$E_j(n)$
in the definition of 
$\xi_n$ 
is replaced by the eigenvalues 
$E_j^C(n)$, $E_j^D (n)$
of 
$H_C$, 
$H_D$
respectively.
By a 
localization argument, for an bounded interval 
$I (\subset {\bf R})$, 
we can find intervals 
$I', I''$ 
such that 
$I' = I - {\cal O}
\left(
\exp \left[ - (const.) n^{\delta \gamma} \right]
\right)$, 
$I'' = I + {\cal O}
\left(
\exp \left[ - (const.) n^{\delta \gamma} \right]
\right)$
with 
\begin{equation}
\xi_n^C
(I' \times [0,1])
\le
\xi_n (I \times B)
\le
\xi_n^D
(I'' \times [0,1]).
\label{localizationestimate}
\end{equation}
In fact, 
as is discussed in \cite{N1, GK} for instance, 
for each eigenvalues of 
$H_C$ 
in 
$I$, 
by smoothing argument near the boundary, 
the corrresponding eigenfunction becomes an approximate eigenfunction of 
$H_n$
so that 
$H_n$
has eigenvalues in 
$I + 
{\cal O}
\left(
\exp \left[ - (const.) n^{\delta \gamma} \right]
\right)$
with those localization centers in 
$B$.
Moereover, 
for each eigenvalues of 
$H_n$ 
in 
$I$ 
localized in 
$B$, 
by cutting off argument we can construct approximate eigenfunctions of 
$H^D$
with eigenvalues in 
$I + 
{\cal O}
\left(
\exp \left[ - (const.) n^{\delta \gamma} \right]
\right)$. 
On the other hand, 
for the eigenvalue process of 
$H_C$, $H_D$, 
we have the Poisson statistics as for 
$H_n$ 
proved in 
\cite{KN2}.
That is, 
the point processes 
$\eta_n^C, \eta_n^D$ 
whose atoms are composed of the rescaled eigenvalues of 
$H_C, H_D$
respectively converge to Poisson 
$\left( |B| d \lambda / \pi \right)$.
In fact, 
the key to the proof for 
$H_n$ 
is that the jump point of the processes 
$\left\lfloor
\Theta_{t}(\lambda)^{(n)}(c) / \pi
\right\rfloor$ 
and 
$\left\lfloor
\Theta_{t}(\lambda)^{(n)}(c') - \Theta_{t}(\lambda)^{(n)}(c) / \pi
\right\rfloor$
converge to Poisson processes and they are asymptotically independent(Proposition 5.7, Remark 5.1 and Lemma 5.11 in \cite{KN2}).
And 
we can show the same statement for the processes 
$\Theta^{(\sharp, n)}_{t}(\lambda)$, 
$\sharp = C, D$ 
where the starting time 
$0$
in 
$\Theta_{t}(\lambda)^{(n)}$
is replaced by 
$an \pm n^{\delta}$
which satisfy the same 
SDE
as for 
$\Theta^{(n)}_{t}(\lambda)$. 
Then 
we can show

%
\begin{lemma}
For 
$\sharp = C, D$, 
\beq
&(1)& \quad
\lim_n
P( \xi_n^{\sharp} (I \times [0,1]) = 0)
=
P( \xi (I \times B) = 0)
\\
&(2)& \quad
\lim_n
{\bf E}[ \xi^{\sharp}_n (I \times [0,1]) ]
=
\frac {|I| \cdot |B|}{\pi} 
=
{\bf E}[ \xi (I \times B)]. 
\eeq
\end{lemma}
%
%
Now, letting 
$n \to \infty$ 
in 
\beq
&&
P(\xi_n^D (I'' \times [0,1]) = 0)
\le
P (\xi_n (I \times B) = 0)
\le
P( \xi_n^C (I' \times [0,1]) = 0), 
\eeq
we have 
\beq
\lim_n
P(\xi_n (I \times B) = 0)
=
P( \xi (I \times B) = 0)
\eeq
and similarly we have 
\begin{equation}
\lim_n
{\bf E}[ \xi_n (I \times B) ]
=
{\bf E}[ \xi (I \times B)] 
=
\frac {|I| \cdot |B|}{\pi}
\label{Poissonintensity}
\end{equation}
yielding Proposition \ref{locctr}(1), (2).
Therefore by 
\cite{K} Theorem 4.7, 
$\xi_n$
converges to a Poisson process whose intensity measure is equal to 
$d \lambda/\pi \times 1_{[0,1]}(x)dx$.
\QED\\
{\bf Remark 1}\\
It is sufficient to show 
$\limsup
{\bf E}[ \xi_n I \times B ]
\le
{\bf E}[ \xi  \times B]$
for the proof of Proposition \ref{locctr} but we will need equality later for the proof of Theorem 1.1. \\
{\bf Remark 2}\\
The argument of 
proof of Proposition \ref{locctr} is almost model-independent.
If 
(i)
eigenfunctions are exponentially localized, and 
(ii)
if any subsystem of size 
${\cal O}(n)$ 
we have Poisson statistics for the eigenvalue process, 
then we have the Poisson convergence for the pairs of eigenvalues and localization centers. 
\\

Theorem 1.2 for 
$\alpha < 1/2$
follows easily from Proposition \ref{locctr}. 
\\
{\it Proof of Theorem 1.2 for 
$\alpha <1/2$}\\
As the other cases, we show 
%
$
\int F(\lambda, \mu) d \Xi^{(n)}(\lambda,\mu)
\stackrel{d}{\to}
\int F(\lambda, \mu) d \Xi(\lambda,\mu)
$
%
for 
$F \in C_c ({\bf R} \times {\cal P}(0,1))$.
By 
Proposition \ref{locctr}, we can assume that these atoms 
$\left(
n( \sqrt{E_j} - \sqrt{E_0}) , 
x_j / n
\right)$
of 
$\xi_n$
converges in distribution to those 
$(P_j, \widetilde{P}_j)$
of 
$\xi$.
Then by Lemma \ref{vague} below, 
we have 
\beq
\mu_{E_j(n)}^{(n)}
\stackrel{d}{\to}
\delta_{ \widetilde{P}_j }
\eeq
so that for 
$F \in C_c ({\bf R} \times {\cal P}(0,1))$, 
\beq
\sum_j
F 
\left(
n \left(
\sqrt{E_j(n)} - \sqrt{E_0}
\right), 
\mu_{E_j(n)}^{(n)}
\right)
\stackrel{d}{\to}
\sum_j
F 
\left(
P_j, 
\delta_{\widetilde{P}_j}
\right).
\eeq
For the intensity measure of 
$\Xi$, 
we use the fact that 
$\{ (P_j, \widetilde{P}_j ) \}_j
\sim
Poisson (d \lambda / \pi \times 1_{[0,1]}(x)dx)$
and yield 
\beq
{\bf E}
\left[
\int F(\lambda, \mu) d \Xi_E
\right]
&=&
{\bf E}
\left[
\sum_j
F(P_j, \delta_{\widetilde{P}_j}
)
\right]
=
\int d \lambda 
{\bf E}[ F(\lambda, \delta_{U} ], 
\quad
U \sim unif\, [0,1].
\eeq
\QED\\
%

%
%
\begin{lemma}
\label{vague}
Suppose
$\left(
n( \sqrt{E_j} - \sqrt{E_0}), x_j / n
\right)
\stackrel{d}{\to}
(P_j, \widetilde{P}_j)$.
Then we have 
\beq
\mu_{E_j(n)}^{(n)}
\stackrel{d}{\to}
\delta_{ \widetilde{P}_j }. 
\quad
\eeq
\end{lemma}
%
%
\begin{proof}
Let 
$\xi = \sum_j \delta_{(P_j, \widetilde{P}_j)}$, 
$(P_j, \widetilde{P}_j) \in {\bf R} \times [0,1]$
be a Poisson process in Proposition \ref{locctr}.
Then by Skorohod's theorem 
we may assume 
\begin{equation}
\left(
n \left(
\sqrt{E_j(n)} - \sqrt{E_0}
\right), 
\frac {
x_{\psi_j}^{(n)}
}
{n}
\right) 
\;
\stackrel{a.s.}{\to}
\;
\left(
P_j, \widetilde{P}_j
\right). 
\label{Skorohod}
\end{equation}
For simplicity, 
we write 
$x_j := x_{\psi_j}^{(n)}$.
We 
represent the normalized eigenfunction 
$\psi_j^{(n)}$
around its localization center 
$x_j$ 
as 
%
$
\psi_j^{(n)}(x)
=:
g_j^{(n)}(x - x_j)
$.
%
For 
$f \in C_c({\bf R})$, 
\beq
\int_{\bf R}
f(x) | \psi_j^{(n)}(nx) |^2 ndx
&=&
\int_{\bf R}
f
\left(
\frac {x_j}{n} + \frac yn
\right) | g_j^{(n)}(y) |^2 dy
\quad
(y = nx - x_j).
\eeq
Here 
we use the following estimate : for any 
$\epsilon > 0$
we can find 
$R_{\epsilon} > 0$
such that
\beq
\int_{ |y| \ge R_{\epsilon} }
| g_j^{(n)} (y) |^2 dy
<
\epsilon.
\eeq
In fact, 
since we only consider eigenvalues 
$E_j (n)$ 
in the 
${\cal O}(n^{-1})$-
neighborhood of 
$E_0$, 
we can bound 
$|g_j^{(n)}(x)| 
\le 
C(E_0)
\exp 
\left[
- D(E_0) |x|^{1 - 2 \alpha}
\right]$
so that 
$R_{\epsilon}$
can be taken uniformly w.r.t. 
$n, j$. 
We then have 
\beq
\left|
\int_{\bf R}
f(x) | \psi_j^{(n)}(nx) |^2 ndx
-
f( \widetilde{P}_j )
\right|
&=&
\left|
\int_{\bf R}
\left\{
f\left(
\frac {x_j}{n} + \frac yn
\right) 
- f( \widetilde{P}_j )
\right\}
| g_j^{(n)}(y) |^2 dy
\right|
\\
& \le &
\sup_{ |y| \le R_{\epsilon} }
\left|
f\left(
\frac {x_j}{n} + \frac yn
\right) 
- f( \widetilde{P}_j )
\right|
+
2 \| f \|_{\infty}
\epsilon.
\eeq
Since 
$f$
is uniformly continuous, 
\beq
\sup_{ |y| \le R_{\epsilon} }
\left|
f\left(
\frac {x_j}{n} + \frac yn
\right) 
- f( \widetilde{P}_j )
\right|
& \le &
\sup_{ |y| \le R_{\epsilon} }
\left|
f\left(
\frac {x_j}{n} + \frac yn
\right) 
- f
\left( 
\frac {x_j}{n} 
\right)
\right|
+
\sup_{ |y| \le R_{\epsilon} }
\left|
f\left(
\frac {x_j}{n}
\right) 
- f( \widetilde{P}_j )
\right|
\\
& \stackrel{n \to \infty}{\to} & 0.
\eeq
\QED
\end{proof}
%

%
\section{Global version}
The result 
for the global version (the statement in Theorem 1.1) follows from that for local version (statement in Theorem 1.2) by a general argument.
Following 
\cite{RV}, 
we introduce
\beq
g_1 (x)
&:=&
(1 - |x|) 1(|x| \le 1), 
\\
G_n(E)
&:=&
\sum_{E_j(n) \in J}
g_1 
\left(
n 
\left(
\sqrt{E_j(n)} - \sqrt{E}
\right)
+ \theta
\right)
\cdot
g_2 \left(
E_j(n), \mu^{(n)}_{E_j(n)}
\right)
\eeq
where 
$g_2 \in C_b ({\bf R} \times {\cal P}(0,1))$.
$\theta \sim unif [0, \pi)$ 
for 
$\alpha > \frac 12$, 
and 
$\theta = 0$ 
otherwise. 
For the global version, 
we need to consider 
$\sum_{E_j \in J}
g_2 \left(
E_j(n), \mu^{(n)}_{E_j(n)}
\right)$. 
The motivation 
to consider 
$G_n (E)$ 
is to localize this quantity around the reference energy 
$E$ 
by multiplying 
$g_1$.
We compute 
$\int_J dN(E)/N(J) \int d {\bf P}G_n (E)$
in two ways by exchanging the order of integrals, and then equate them by the Fubini theorem, which yields the conclusion. 
The idea 
behind this argument is : 

(1)
Integrate w.r.t. 
$dN(E)/N(J)$
and then take expectation : 
we first 
integrate w.r.t. the reference energy around each 
$E_j$'s 
of width of order 
${\cal O}(n^{-1})$ 
which results in to get 
$n^{-1}$
factor, yielding the quantity we want to compute. 
\\
(2)
Take expectation first and then integrate w.r.t. 
$dN(E)/N(J)$ : 
we first fix the reference energy, and take expectation first.
Since 
we have 
$g_1$ 
factor, we have the intensity measure of the local version.
Then 
integrating w.r.t. the reference energy gives us the answer.\\
~Therefore, 
a general principle is that,  
the answer to our global problem is equal to the integral w.r.t. the reference energy of the intensity measure of the local problem.\\
Along the idea 
explained above, we compute 
$\int_J d N(E)/N(J){\bf E}[ G_L (E) ]$, 
$J=[a,b]$
in two ways. 
We first note that 
%
$
| n(\sqrt{E_j} - \sqrt{E}) + \theta |
\le 1
$
if and only if 
$
| \sqrt{E_j} - \sqrt{E} |
\le
(\pi+1)/n
$.
Since
$\sqrt{E_j} \in (\sqrt{a}, \sqrt{b})$, 
$a>0$, 
we have 
\beq
\int_J
g_1 (n 
\left(
\sqrt{E_j} - \sqrt{E}
\right) + \theta)
d N(E)
=
\frac {1}{n \pi}.
\eeq
(1)
We first integrate w.r.t. 
$dN(E)$
and then take expectation : 
\begin{eqnarray}
&&
{\bf E}\left[ 
\int_J \dfrac {dN(E)}{N(J)} G_n (E) 
\right]
\nonumber
\\
&=&
{\bf E}
\left[
\frac {1}{N(J)}
\frac {1}{\pi n}
\sum_{E_j(n) \in J}
g_2 \left(
E_j(n), \mu^{(n)}_{E_j(n)}
\right)
\right]
\nonumber
\\
&=&
{\bf E}
\left[
\frac {1}{
N (H_n, J)
}
\cdot
\frac {1}{\pi}
\cdot
\sum_{E_j(n) \in J}
g_2 \left(
E_j(n), \mu^{(n)}_{E_j(n)}
\right)
\right]
+ o(1).
\label{first}
\end{eqnarray}
where we set 
$N (H_n, J)
:=
\sharp \{
\mbox{ eigenvalues of $H_n$ in $J$ } \}$.
The 
last equality follows from 
\begin{equation}
{\bf E}
\left[
\left(
\frac {1}{N (H_n, J) }
-
\frac {1}{N(J) n}
\right)
\frac {1}{\pi}
\sum_{E_j \in J}
g_2 (E_j, \mu^{(n)}_{E_j(n)})
\right]
= o(1). 
\label{error}
\end{equation}
To show 
(\ref{error}), 
we note that the quantity in the expectation is estimated as  
\beq
&&
\left|
\frac {1}{N (H_n, J)}
-
\frac {1}{N(J) n}
\right|
\frac {1}{\pi}
\sum_{E_j \in J}
| g_2 (E_j, \mu^{(n)}_{E_j(n)}) |
\\
&\le&
\left|
\frac {1}{N (H_n, J)}
-
\frac {1}{N(J) n}
\right|
\frac {1}{\pi}
\| g_2 \|_{\infty}
N (H_n, J) 
\\
&=&
\frac {
\bigl|
N(J) - 
n^{-1}N (H_n, J) 
\bigr|
}
{
N(J) 
}
\frac {1}{\pi}
\| g_2 \|_{\infty}
\eeq
which converges to 
$0$
a.s.
by the definition of 
$N(J)$.
On the other hand, since we have 
\beq
N (H_n, J)
\le
\frac {
\theta_n (\sqrt{b}) - \theta_n (\sqrt{a})
}
{\pi}
\eeq
and since, by examining the integral equation 
(\ref{Prufer1}) 
satisfied by 
$\widetilde{\theta}_t (\kappa)$, 
we have 
\beq
\theta_n (\sqrt{b})- \theta_n (\sqrt{a})
\le
C n
\eeq
for some deterministic constant 
$C$, 
$(\ref{error})$
follows from the bounded convergence theorem.
\\

(2)
We first 
expectation and then integrate w.r.t. 
$dN(E)$ : 
\beq
&&
\frac {1}{N(J)}
\int_J dN(E)
{\bf E}[ G_n (E) ]
\\
&=&
\frac {1}{N(J)}
\int_J dN(E)
{\bf E}
\left[
\sum_{E_j \in J}
g_1 
\left(
n
\left(
\sqrt{E_j} - \sqrt{E}
\right) + \theta
\right)
g_2 (E_j, \mu_{E_j}^{(n)})
\right]
\\
&=&
\frac {1}{N(J)}
\int_J dN(E)
{\bf E}
\left[
\sum_{E_j \in J}
g_1 
\left(
n
\left(
\sqrt{E_j} - \sqrt{E}
\right) + \theta
\right)
\left(
g_2 (E, \mu_{E_j}^{(n)}) + o(1)
\right)
\right]
\\
&=&
\frac {1}{N(J)}
\int_J dN(E)
{\bf E}
\left[
\sum_{E_j \in J}
g_1 
\left(
n
\left(
\sqrt{E_j} - \sqrt{E}
\right) + \theta
\right)
\left(
g_2 (E, \mu_{E_j}^{(n)})
\right)
\right]
 + o(1)
\\
&=&
\frac {1}{N(J)}
\int_J dN(E)
{\bf E}
\left[
\int
g_1(\lambda+ \theta) g_2(E, \mu)
d \Xi^{(n)}_E
(\lambda, \mu)
\right]
+ o(1). 
\eeq
where, in the second equality, we used that fact that 
$E_j (n) \to E$ 
for 
$j$'s 
such that 
$n
\left(
\sqrt{E_j} - \sqrt{E}
\right) + \theta
\in
\mbox{ supp } g_1$.
For the third equality, 
we used the fact that
\beq
\sum_{E_j \in J}
g_1 
\left(
n
\left(
\sqrt{E_j} - \sqrt{E}
\right) + \theta
\right)
g_2 (E_j, \mu_{E_j}^{(n)}), 
\quad
\sum_{E_j \in J}
g_1 
\left(
n
\left(
\sqrt{E_j} - \sqrt{E}
\right) + \theta
\right)
g_2 (E, \mu_{E_j}^{(n)})
\eeq
are both uniformly integrable w.r.t. 
$d N(E) \times {\bf P}$ 
by the argument in Section 2 in Appendix.\\
Since 
$\Xi^{(n)}_E \stackrel{d}{\to} \Xi_E$
by Theorem 1.2, and since 
$\{ G_n (E) \}_n$
is uniformly integrable w.r.t. 
$d N \times {\bf P}$
to be shown in Section 2 in Appendix, we have 
\begin{eqnarray}
&&
\int_J \dfrac {dN(E)}{N(J)} {\bf E}[ G_n (E) ]
\nonumber
\\
& \stackrel{n \to \infty}{\to} &
\int_J \dfrac {dN(E)}{N(J)} {\bf E}
\left[
\int g_1 (\lambda+ \theta) 
g_2 (E, \mu)
d \Xi_E(\lambda, \mu)
\right]
\nonumber
\\
&=&
\int_J \dfrac {dN(E)}{N(J)}
\frac {1}{\pi}
\left\{
\begin{array}{ll}
\int d \lambda 
g_1 (\lambda)
{\bf E}
\left[
g_2
\left(
E, 
1_{[0,1]} (t) dt 
\right)
\right]
& 
(\alpha > 1/2) \\
\int_J d \lambda
g_1 (\lambda)
{\bf E}
\left[
g_2
\left(
E, 
\frac {
\exp
\Bigl(
2 {\cal Z}_{
\tau(E) \log \frac tU
}
-
2 \tau(E)
\log 
\left| \frac tU \right|
\Bigr)
dt
}
{
\int_0^1
\exp
\Bigl(
2 {\cal Z}_{
\tau(E) \log \frac sU
}
-
2 \tau(E)
\log 
\left| \frac sU \right|
\Bigr)
ds
}
\right)
\right]
&
(\alpha = 1/2)
\\
\int d \lambda
g_1 (\lambda)
{\bf E}
\left[
g_2
\left(
E, 
\delta_{ U }
\right)
\right]
& 
(\alpha > 1/2) 
\end{array}
\right.
\nonumber
\\
&=&
\int_J \dfrac {dN(E)}{N(J)}
\frac {1}{\pi}
\left\{
\begin{array}{ll}
{\bf E}
\left[
g_2
\left(
E, 
1_{[0,1]} (t) dt 
\right)
\right]
& 
(\alpha > 1/2) \\
{\bf E}
\left[
g_2
\left(
E, 
\frac {
\exp
\Bigl(
2 {\cal Z}_{
\tau(E) \log \frac tU
}
-
2 \tau(E)
\log 
\left| \frac tU \right|
\Bigr)
dt
}
{
\int_0^1
\exp
\Bigl(
2 {\cal Z}_{
\tau(E) \log \frac sU
}
-
2 \tau(E)
\log 
\left| \frac sU \right|
\Bigr)
ds
}
\right)
\right]
&
(\alpha = 1/2)
\\
{\bf E}
\left[
g_2
\left(
E, 
\delta_{ U }
\right)
\right]
& 
(\alpha > 1/2) 
\end{array}
\right.
\label{second}
\end{eqnarray}
(\ref{first}), (\ref{second})
yield the statement of Theorem 1.1.
\QED
%

%
\section{Appendix}
%
%
\subsection{Statement for DC model}
We  
consider the continuum 
1-dimensional operator with decaying coupling constant,
that is, 
\beq
H_{\alpha, n}
=
- \frac {d^2}{d t^2} + n^{- \alpha} F(X_t), 
\quad
on 
\;
L^2 (0, n).
\eeq
with Dirichlet boundary condition.
Then 
we have the corrresponding results for theorems 1.1, 1.2 which we state here. 
The conclusion 
for the critical case is essentially the same as that in 
Rifkind-Virag \cite{RV}.
Since 
the proof is similar for those of theorems 1.1, 1.2, we omit details. 
\begin{theorem}
\beq
&&
\left(
E_J^{(n)}, \mu^{(n)}_{E_J^{(n)}}
\right)
\\
&&
\stackrel{d}{\to}
\left\{
\begin{array}{ll}
\left(
E_J, 1_{[0,1]}(t) dt 
\right) & (\alpha > 1/2) \\
\left(
E_J, 
\frac {
\exp
\Bigl(
2 {\cal Z}_{\tau(E_J)(t - U)} - 2  
\left|
\tau(E_J)(t - U) 
\right|
\Bigr)
dt
}
{
\int_0^1
\exp
\Bigl(
2 {\cal Z}_{\tau(E_J)(s - U)} - 2  
\left|
\tau(E_J)(s - U)
\right|
\Bigr)
ds
}
\right)
& (\alpha=1/2) \\
\left(
E_J, \delta_{unif[0,1]}(dt)
\right)
& (\alpha < 1/2) 
\end{array}
\right.
\eeq
\end{theorem}
For the 
Local version, we set 
\beq
\Xi^{(n)}_{E_0}
:=
\sum_j
\delta_{
\Bigl(
n 
\bigl(
\sqrt{E_j(n)} - \sqrt{E_0}
\bigr)+ \theta, 
\,
\mu_{E_j(n)}^{(n)}
\Bigr)
}, 
\quad
\eeq
where 
$\theta \sim unif [0, \pi)$
for 
$\alpha \ge \frac 12$, 
$\theta = 0$ 
for 
$\alpha < \frac 12$.
%
\begin{theorem}
$\Xi^{(n)}_{E_0} \stackrel{d}{\to} \Xi_{E_0}$, 
where
\beq
\Xi_{E_0}
&=&
\left\{
\begin{array}{ll}
\sum_{j \in {\bf Z}}
\delta_{j \pi + \theta}
\otimes
\delta_{ 1_{[0,1]}(t) dt }
& (\alpha > 1/2) \\
\sum_{ \lambda : Sch^* }
\delta_{
\lambda
}
\otimes
\delta
\Bigl(
\frac { 
\exp ( 2 \widetilde{r}_t(\lambda) ) dt
}
{
\int_0^1
\exp ( 2 \widetilde{r}_s(\lambda) ) ds
}
\Bigr)
& (\alpha = 1/2) \\
\sum_{j \in {\bf Z}}
\delta_{P_j}
\otimes
\delta_{\widetilde{P}_j}, 
\quad
&
(\alpha < 1/2) 
\end{array}
\right.
\eeq
where for 
$\alpha > 1/2$, 
$\theta \sim unif [0, \pi)$. 
For
$\alpha = 1/2$, 
\beq
Sch^*
:=
\left\{
\lambda \in {\bf R}
\, \middle| \,
\Psi_1(\lambda) \in 2j \pi + unif [0, 2\pi), 
\;
j \in {\bf Z} 
\right\}, 
\eeq
and
$\{ \Psi_t (\lambda)\}$
is a increasing function valued process given in eq.(1.2) in \cite{N2}.
$\widetilde{r}_t (\lambda)$
is characterized by the solution to the following equation : 
\beq
d 
\widetilde{r}_t (\lambda)
=
\tau(E_0)
dt
+
\sqrt{
\tau(E_0)
}
\,
d B_t^{\lambda}, 
\quad
t > 0
\eeq
where 
$\{ B_t^{\lambda} \}_{\lambda}$
is a family of Brownian motion.
For 
$\alpha < 1/2$, 
$\{ P_j \} : Poisson (d \lambda /\pi)$, 
$\{ \widetilde{P}_j \} : Poisson (1_{[0,1]}(t) dt)$
where 
$Poisson (\mu)$
is the Poisson process with intensity measure 
$\mu$. 
The intensity measure of 
$\Xi_{E_0}$ 
is given by 
\beq
&&
{\bf E} 
\left[
\int
G(\lambda, \nu) 
d \Xi_{E_0}(\lambda, \nu)
\right]
\\
&=&
\frac {1}{\pi}
\left\{
\begin{array}{ll}
\int d \lambda
{\bf E}
\left[
G
\left(
\lambda, 
1_{[0,1]} (t) dt 
\right)
\right]
& 
(\alpha > 1/2) \\
\int d \lambda
{\bf E}
\left[
G
\left(
\lambda, 
\frac {
\exp
\Bigl(
2 {\cal Z}_{
\tau(E_0) (t-U)
}
-
2 
\tau(E_0) \left| t-U \right|
\Bigr)
dt
}
{
\int_0^1
\exp
\Bigl(
2 {\cal Z}_{
\tau(E_0) (s-U)
}
-
2 
\tau(E_0)|s-U |
\Bigr)
ds
}
\right)
\right]
&
(\alpha = 1/2)
\\
\int d \lambda
{\bf E}
\left[
G
\left(
\lambda, 
\delta_{ U }
\right)
\right]
& 
(\alpha > 1/2) 
\end{array}
\right.
\eeq
where 
$U := unif [0,1]$.
\end{theorem}
%

%
\subsection{Uniform integrability}
In this subsection 
we show the uniform integrability of 
$G_n (E)$
w.r.t. 
$dN \times P$. 
Since 
supp $g_1 \subset \{ |\lambda| \le 1 \}$, 
by setting 
$N (H_n, J)
:=
\sharp \{
\mbox{ eigenvalues of $H_n$ in $J$ } \}$, 
we have for 
$c \ge 1$, 
\begin{eqnarray}
| G_n (E) |
& \le &
\| g_1 \|_{\infty}
\| g_2 \|_{\infty}
N 
\left(
H_n, 
\sqrt{E_0}
+
\frac 1n
(-c,c)
\right)
\nonumber
\\
& \le &
\| g_1 \|_{\infty}
\| g_2 \|_{\infty}
\frac {1}{\pi}
\left(
\Theta_t^{(n)} (c)
-
\Theta_t^{(n)} (-c)
\right)
\label{suffice}
\end{eqnarray}
so that it suffices to show the uniform integrability of 
$\left\{
\Theta_t^{(n)} (c)
\right\}_n$
w.r.t. 
$d N \times P$, 
which in turn follows from either one of the following two statements. 
\begin{eqnarray}
&(1)& \quad
\int_J d N(E)
{\bf E}[ 
\Theta_t^{(n)} (c)
]
\to
\int dN(E)
{\bf E}[ 
\Theta_t (c)
]
\label{uint1}
\\
&(2)& \quad
\sup_n 
\int_J d N(E) 
{\bf E}
\left[ 
\Theta_t^{(n)} (c)^{1+\delta}
\right] 
\quad
\mbox{ for some }
\;
\delta > 0
\label{uint2}
\end{eqnarray}
where we note that 
$\Theta_t (c) \ge 0$
for 
$c \ge 0$.
We shall show 
(\ref{uint1})
or
(\ref{uint2})
in Section 4.2.1, 4.2.2 
for super-critical and critical cases respectively.
For 
sub-critical case we can show the uniform integrability directly to be done in Section 4.2.3.
%

%
\subsubsection{
Supercritical case
}
We show 
(\ref{uint2})
in 
super-critical case. 
By definition, 
\beq
\Theta_t^{(n)}(c)
&=&
ct
+
\widetilde{\theta}_{nt} (\kappa_c)
-
\widetilde{\theta}_{nt} (\kappa)
\eeq
and we write 
$\kappa := \kappa_0$
in this section.
By the integral equation 
eq.(\ref{Prufer1}) 
satisfied by 
$\widetilde{\theta}_t (\kappa_c)$, 
\beq
\widetilde{\theta}_{nt} (\kappa_c)
-
\widetilde{\theta}_{nt} (\kappa)
&=&
\frac {1}{2 \kappa}
Re 
\left(
J^{(n)}_t (\kappa_c)
-
J^{(n)}_t (\kappa)
\right)
\\
&& +
\frac {
- 2 \cdot \frac cn
}
{
2 \kappa_c \cdot 2 \kappa
}
\int_0^{nt}
Re
\left(
e^{2i \theta_s(\kappa_c)}- 1 
\right)
a(s) F(X_s) ds.
\\
\mbox{where }
\quad
J^{(n)}_t (\kappa_c)
&:=&
\int_0^{nt}
a(s) 
e^{2i \theta_s (\kappa_c)}
F(X_s) ds
\eeq
Second term goes to 
$0$
as 
$n \to \infty$ 
uniformly w.r.t. 
$(\kappa, \omega)$, 
so that it suffices to show the uniform integrability of 
$J_t^{(n)}(\kappa_c)$
for any 
$c \ge 0$.
We use ``Ito's formula" 
\beq
e^{2i \kappa s} F(X_s) ds
=
d \left(
e^{2i \kappa s} g_{\kappa} (X_s) 
\right)
-
e^{2i \kappa s} 
\nabla g_{\kappa} (X_s) 
d X_s
\eeq
and compute the integral by parts : 
\beq
J^{(n)}_t (\kappa_c)
&=&
\left[
a(s) 
e^{2i \theta_s(\kappa_c)}
g_{\kappa} (X_s) 
\right]_0^{nt}
\\
&&
-
\int_0^{nt}
a'(s) e^{2i \theta_s(\kappa_c)}
g_{\kappa} (X_s) ds
\\
&&
-
\frac {2i}{2 \kappa_c}
\int_0^{nt}
Re
\left(
e^{2i \theta_s(\kappa_c)} - 1
\right)
e^{2i \theta_s(\kappa_c)}
a(s)^2 F(X_s)
g_{\kappa} (X_s) ds
\\
&&
- 2i \cdot \frac cn
\int_0^{nt} 
a(s)
e^{2i \theta_s(\kappa_c)}
g_{\kappa} (X_s) ds
\\
&&
-
\int_0^{nt}
a(s)
e^{2i \theta_s(\kappa_c)}
\nabla g_{\kappa} (X_s) 
dX_s
\\
&=:&
J_1 + \cdots + J_5.
\eeq
Here 
we use the notation 
${\cal O}(1)$
if the quantity in question is uniformly bounded w.r.t. 
$(\kappa, \omega) \in J \times \Omega$.
Then we have 
\beq
J_1
&=&
{\cal O}(1)
\\
|J_2|
& \le &
\int_0^{nt} a'(s)
\left|
e^{2i \theta_s(\kappa_c)} g_{\kappa} (X_s) 
\right| 
ds
\le 
(Const.)
\int_0^{nt} a'(s) ds
=
{\cal O}(1)
\\
|J_3|
& \le &
(Const.)
\int_0^{nt} a(s)^2 ds
=
{\cal O}(1)
\\
|J_4|
&\le &
(Const.)
\frac 1n
\int_0^{nt} a(s)ds
=
{\cal O}(n^{-\alpha})
\\
\langle 
| J_5 |^2
\rangle
& \le &
(Const.)
\int_0^{nt}
a(s)^2 ds
=
{\cal O}(1). 
\eeq
Getting together we have 
\beq
\sup_n
\int_J dN (E)
{\bf E}
\left[
| J_t^{(n)} (\kappa_c) |^2
\right]
<
\infty. 
\eeq
%

\subsubsection{Critical case}
We show 
(\ref{uint1})
for the critical case.
In fact,
\cite{KN1} Lemma 6.3 says
\begin{eqnarray}
\Theta_t^{(n)}(c)
&=&
2 ct 
+ 
Re\; \epsilon_t^{(n)}
+
\frac {1}{\kappa} 
Re \, V_t^{(n)}(c)
+
\frac {1}{\kappa}
Re
\left(
\delta_{nt}(\kappa_c) - \delta_{nt}(\kappa)
\right)
\nonumber
\\
\mbox{ where }
\quad
|\epsilon_t^{(n)}|
& \le &
\frac Cn
\int_0^{nt} a(s) ds
\stackrel{n \to \infty}{\to} 0
\nonumber
\\
V_t^{(n)}(c)
&:& 
\;
\mbox{ Martingale so that } 
\quad
{\bf E}[ V_t^{(n)} (c) ] = 0
\nonumber
\\
&{\bf E}&
\left[
\max_{0 \le t \le T}
\left|
\delta_{nt}(\kappa_c) - \delta_{nt}(\kappa)
\right|^2
\right]
\stackrel{n \to \infty}{\to} 0
\label{delta}
\end{eqnarray}
which implies 
%
$
{\bf E}[ \Theta_1^{(n)} (c) ] 
\stackrel{n \to \infty}{\to}
{\bf E}[ \Theta_1 (c) ] = 2 ct.
$
%
Moreover, 
by examining the proof of Lemma 6.3 in \cite{KN1}, 
the LHS of 
(\ref{delta})
is locally bounded w.r.t. 
$E$, 
and so is 
${\bf E}[ \Theta_1^{(n)}(c) ]$.
By the bounded convergence theorem, we now have 
\beq
\int_J d N(E)
{\bf E}[ \Theta_1^{(n)} (c) ] 
\to
\int_J d N(E)
{\bf E}[ \Theta_1 (c) ] 
= 2 ct. 
\eeq
%
\subsubsection{Sub-critical case}
We show 
the uniform integrability directly for the sub-cricial case.
By (\ref{suffice}), 
it suffices 
to show the uniform integrability of 
$N 
\left(
H_n, 
\sqrt{E_0}
+
\frac 1n
(-c,c)
\right)
=
\xi_n ((-c, c) \times [0,1])$.
In Section 2, 
it has been shown that 
%
$
\lim_n
{\bf E}[ 
\xi_n (I \times B)
]
=
{\bf E}[ 
\xi (I \times B)
].
$
%
The quantities in LHS are all locally bounded for 
$E$.
In fact, 
$\xi_n (I \times B)$ 
is governed by the number of jump points of 
$t \mapsto \left\lfloor
\Theta_{n t}(c)/ \pi 
\right\rfloor$, 
and the SDE satisfied by 
$\Theta_t (c)$
is determined by 
$E$
and 
$\langle F g_{\sqrt{E}} \rangle$
only, and 
$\langle F g_{\sqrt{E}} \rangle$
is bounded for 
$E \in J$.
Here 
we used the condition that the left-end 
$a$ 
of the interval 
$J$ 
is positive.
Then 
by the bounded convergence theorem, we have 
\beq
\lim_{n \to \infty}
\int_J d N(E)
{\bf E}[ \xi_n I \times B ]
=
\int_J d N(E)
{\bf E}[ \xi I \times B], 
\quad
J = [a,b]. 
\eeq
%

\vspace*{1em}
\noindent {\bf Acknowledgement }
This work is partially supported by 
JSPS KAKENHI Grant 
Number 20K03659(F.N.).

%

\end{document}